\documentclass[11pt,journal,twocolumn]{IEEEtran}
\usepackage[boxed]{algorithm2e}
\SetKwFor{Alg}{Algorithm}{}{}
\SetKwFor{Init}{Initialization}{}{}
\SetKwFor{For}{for}{}{}
\SetKwFor{While}{while}{}{}
\SetKwIF{If}{ElseIf}{Else}{if}{}{else if}{else}{} 
\usepackage{stfloats}
\usepackage{url}
\usepackage[utf8]{inputenc}
\usepackage[T1]{fontenc}
\usepackage{amsmath,amssymb,amsfonts,latexsym,amsthm,enumerate,fullpage,hyperref,dsfont}
\usepackage{stackrel,color,mathtools}
\usepackage[noadjust]{cite}
\usepackage{graphicx}
\usepackage{slashbox}
\usepackage{setspace}
\mathtoolsset{showonlyrefs=true}
\allowdisplaybreaks 
\newtheorem{theorem}{Theorem}
\newtheorem{lemma}{Lemma}
\newtheorem{definition}{Definition}
\newtheorem{remark}{Remark}
\newtheorem{corollary}{Corollary}

\newcommand{\dfn}{\triangleq}
\newcommand{\bs}[1]{\boldsymbol{#1}}
\newcommand{\mybf}[1]{{\bf #1}}
\newcommand{\indfunc}[1]{\mathds{1}\left(#1\right)}

\def\BSCprob{\varepsilon}
\newcommand{\targetProb}{\delta}
\newcommand{\nrepetitions}{\rho}
\def\Expt{\mathds{E}}
\def\tind{i}
\def\protLength{n}
\def\binarySpace{\mathbb{F}_2}
\def\GFq{\mathbb{F}_q}
\def\GFql{\mathbb{F}_{q_\levelind}}
\newcommand{\transFunc}{\pi}
\newcommand{\protTrans}{\tau}
\newcommand{\channelTrans}{X}
\newcommand{\channelReceive}{Y}

\newcommand{\nlevels}{L}
\newcommand{\levelind}{l}
\newcommand{\levelblock}{k}
\newcommand{\cursor}{j}

\newcommand{\Pfailure}{{P_e}}

\newcommand{\repetitioncoeff}{a}
\newcommand{\repetitioncoeffH}{\tilde{a}}
\newcommand{\ntrans}{T}

\newcommand{\sigmalength}{N}
\newcommand{\Cshannon}{C_{\mathrm{Sh}}}
\newcommand{\Cinter}{C_{\mathrm{I}}}
\newcommand{\simulatingprotocol}{\Sigma}
\newcommand{\Pe}{P_e}

\newcommand{\rewindbit}{b}
\newcommand{\rewindwindow}[1]{w\left[#1\right]} 
\newcommand{\Pbl}{P_{\rewindbit_{\levelind}}}
\newcommand{\Pbone}{P_{\rewindbit_{1}}}
\newcommand{\Pbonebound}{\overline{P}_{\rewindbit_{1}}}
\newcommand{\Pblbound}{\overline{P}_{\rewindbit_{\levelind}}}
\newcommand{\batdistance}{\beta}
\newcommand{\batdistanceBMS}{\bs{\beta}}

\newcommand{\averewind}{{\overline{\rewindbit_{\levelind}}}}

\newcommand{\detectederror}[1]{\boldsymbol{#1}}
\newcommand{\zeroedbits}[1]{\textcolor{blue}{#1}}
\newcommand{\undcodedSim}{\mybf{X}}

\newcommand{\syndromevec}{\mybf{s}}

\newcommand{\noisevec}{\mybf{Z}}
\newcommand{\Pmd}{\gamma}
\newcommand{\RBsc}{R_{BSC}}
\newcommand{\erasure}{\mathsf{E}}
\newcommand{\cursormin}{\cursor}
\newcommand{\m}{\mathcal}
\newcommand{\rewindbitOr}{\rewindbit}
\newcommand{\Ber}{\mathop{\mathrm{Bernoulli}}}

\DeclareMathOperator*{\argmax}{\arg\!\max}

\newcommand{\newtext}[1]{#1}
\newcommand{\newtextrevb}[1]{#1}

\begin{document}
\title{A Lower Bound on the {Essential} Interactive Capacity of Binary Memoryless Symmetric Channels
}
\author{Assaf Ben-Yishai, Young-Han Kim, Or Ordentlich and Ofer Shayevitz
	\thanks{ 
	O.~Ordentlich and A.~Ben-Yishai are with the School of Computer Science and Engineering, Hebrew University of Jerusalem, Israel. 
	O.~Shayevitz is with the Department of EE--Systems, Tel Aviv University, Tel Aviv, Israel. 
	Y.-H.~Kim is with the Department of Electrical and Computer Engineering, University of California, San Diego, La Jolla, CA 92093 USA. 
	Most of this work has been performed while A.~Ben-Yishai was 
	with the Department of EE--Systems, Tel Aviv University. 
	Emails: \{assafbster@gmail.com, yhk@ucsd.edu, or.ordentlich@mail.huji.ac.il, ofersha@eng.tau.ac.il\}
	The work of A. Ben-Yishai was supported by an ISF grant no. 1367/14. 
	The work of O. Orderntlich was supported by an ISF grant no. 1791/17.
	The work of O. Shayevitz was supported by an ISF grant no. 1495/18 and an	ERC grant no. 639573.
	This paper was presented in part at ISIT 2019.}}
\maketitle


\begin{abstract}
\newtext{The essential interactive capacity of a discrete memoryless channel is defined in this paper as the maximal rate at which the transcript of any interactive protocol can be reliably simulated over the channel, using a deterministic coding scheme. In contrast to other interactive capacity definitions in the literature, this definition makes no assumptions on the order of speakers (which can be adaptive) and does not allow any use of private / public randomness; hence, the essential interactive capacity is a function of the channel model only. It is shown that the essential interactive capacity of any binary memoryless symmetric (BMS) channel is at least $0.0302$ its Shannon capacity. To that end, we present a simple coding scheme, based on extended-Hamming codes combined with error detection, that achieves the lower bound in the special case of the  binary symmetric channel (BSC). We then adapt the scheme to the entire family of BMS channels, and show that it achieves the same lower bound using extremes of the Bhattacharyya parameter. }
\end{abstract}


\section{Introduction}
In the classical Shannon one-way communication problem, a transmitter (Alice) wishes to send a message reliably to a receiver (Bob) over a memoryless noisy channel. She does so by mapping her message into a sequence of channel inputs (codeword) in a predetermined way, which is corrupted by the channel and then observed by Bob, who tries to recover the original message. The {\em Shannon capacity} of the channel, which is the maximal number of message bits per channel use that Alice can convey to Bob with vanishingly low error probability, quantifies the most efficient way to do so. In the two-way channel setup \cite{shannon1961two}, both parties draw independent messages and wish to exchange them over a two-input two-output memoryless noisy channel, and the Shannon capacity (region) is defined similarly. Unlike the one-way case, both parties can now employ adaptive coding by incorporating their respective observations of the past channel outputs into their transmission processes. However, just as in the one-way setup, the messages they wish to exchange are determined before communication begins. In other words, if Alice and Bob had been connected by a noiseless bit pipe, they could have simply sent their messages without any regard to the message of their counterpart.  

In a different two-way communication setup, generally referred to as {\em interactive communication}, the latter assumption is no longer held true. In this interactive communication setup, Alice and Bob do not necessarily wish to disclose all their local information. What they want to tell each other depends, just like in human conversation, on what the other would tell them. A simple instructive example (taken from \cite{gelles2015coding}) is the following. Suppose that Alice and Bob play chess remotely, by announcing their moves over a communication channel (using, say, $12$ bits per move, which is clearly sufficient). If the moves are conveyed without error, then both parties can keep track of the state of the board, and the game can proceed to its termination. The sequence of moves occurring over the course of this noiseless game is called a {\em transcript}, and it is dictated by the {\em protocol} of the game, which constitutes Alice and Bob's respective strategies determining their moves at any given state of the board.  

Now, assume that Alice and Bob play the game over a noisy two-way channel, yet wish to simulate the transcript as if no noises were present. In other words, they would like to communicate back and forth in a way that ensures, once communication is over, that the transcript of the noiseless game can be reproduced by to both parties with a small error probability. They would also like to achieve this goal as efficiently as possible, i.e., with the least number of channel uses. One direct way to achieve this is by having both parties describe their entire protocol to their counterpart, i.e., each and every move they might take given each and every possible state of the board. This reduces the interactive problem to a non-interactive one, with the protocol becoming a pair of messages to be exchanged. However, this solution is grossly inefficient; the parties now know much more than they really need in order to simply reconstruct the transcript. At the other extreme, Alice and Bob may choose to describe the transcript itself by encoding each move separately on the fly, using a short error correcting code. Unfortunately, this code must have some fixed error probability and hence an undetected error is bound to occur at some unknown point, causing the states of the board held by the two parties to diverge, and rendering the remainder of the game useless. It is important to note that if Alice and Bob had wanted to play sufficiently many games in parallel, then they could have used a long error-correcting code to simultaneously protect the set of all moves taken at each time point, which in principle would have let them operate at the one-way Shannon capacity (which is the best possible). The crux of the matter therefore lies in the fact that the interactive problem is \textit{one-shot}, namely, only a single instance of the game is being played.

In light of the above, it is perhaps surprising that it is nevertheless possible to simulate any one-shot interactive protocol using a number of channel uses that is proportional to the length of the transcript. In other words, a positive rate of simulation is achievable whenever the Shannon capacity is nonzero. This fact was initially proved by Schulman \cite{schulman1992communication}, who was also the first to introduce the notion of interactive communication over noisy channels. However, this rate of reliable simulation has never been quantified; it is only known to be some nonzero fraction of the Shannon capacity. \newtext{Moreover, several subtly different notions of achievability exist in the literature, depending in particular on various assumptions on the structure of the protocol and on the randomness resources (see Section~\ref{section:previous}). In order to circumvent these issues, we define a stringent notion of achievability that depends only on the channel; in particular, our definition does not make any assumptions on the simulated protocol, and does not allow the use of public or private randomness. We show that the maximal achievable rate under this definition, which we call the \textit{essential interactive capacity}, is at least a $0.0302$ fraction of the Shannon capacity for the entire family of binary memoryless symmetric (BMS)  channels, which includes in particular the binary symmetric channel (BSC). }

The rest of the paper is organized as follows. In Section~\ref{section:formulation} we present the problem formulation and a high level description of the techniques. 
In Section~\ref{sec:maincontrib} we present the main contribution. In Section~\ref{section:previous} we put our work in context of existing results in the literature. We provide some necessary preliminaries in Section~\ref{section:preliminaries}, and then state the main results in Section~\ref{section:mainresults}. The coding scheme used in the proof for the binary symmetric channel (BSC) is  presented and analyzed in Sections~\ref{section:SchemeDescription} and \ref{section:schemeanalysis} respectively, and then generalized to binary memoryless symmetric (BMS) channels in Section~\ref{section:bms}. Finally, in Section~\ref{section:derandomize}, we explain how the randomized coding scheme can be modified to be fully deterministic.

\section{Problem Formulation\label{section:formulation}}

\subsection{\newtext{Interactive Communication and the Essential Interactive Capacity}}
A \textit{length-$n$ interactive protocol} is the triplet 
$\bs{\transFunc}\dfn(\bs{\phi}^\text{Alice},\bs{\phi}^\text{Bob},\bs{\psi})$, where
\begin{align}
\bs{\phi}^\text{Alice} &\dfn \left\{{\phi}^\text{Alice}_{\tind}:\{0,1\}^{i-1}\mapsto \{0,1\}\right\}_{\tind=1}^\protLength\\
\bs{\phi}^\text{Bob} &\dfn \left\{{\phi}^\text{Bob}_{\tind}:\{0,1\}^{i-1}\mapsto \{0,1\}\right\}_{\tind=1}^\protLength\\ 
\bs{\psi} &\dfn \left\{{\psi}_{\tind}:\{0,1\}^{i-1}\mapsto \{\text{Alice},\text{Bob}\}\right\}_{\tind=1}^\protLength.
\end{align}
The functions $\bs{\phi}^\text{Alice}$ are known only to Alice, and the functions $\bs{\phi}^\text{Bob}$ are known only to Bob. The \textit{speaker order functions} $\bs{\psi}$ are known to both parties. The \textit{transcript} $\bs{\protTrans}$ associated with the \textit{input} protocol $\bs{\transFunc}$ is sequentially generated by Alice and Bob as follows
\begin{align}\label{eq:transfunc}
\protTrans_{\tind} &=\begin{cases}
{\phi}^\text{Alice}_{\tind}(\bs{\protTrans}^{\tind-1}) & \sigma_\tind=\text{Alice}\\
{\phi}^\text{Bob}_{\tind}(\bs{\protTrans}^{\tind-1}) & \sigma_\tind=\text{Bob}
\end{cases}
\end{align}
where $\sigma_\tind$ is the identity of the speaker at time $\tind$, which is given by:
\begin{align}\label{eq:transfunc2}
\sigma_\tind&={\psi}_{\tind}(\bs{\protTrans}^{\tind-1}).
\end{align}
In the \textit{interactive simulation problem} Alice and Bob would like to \textit{simulate} the 
transcript $\bs{\protTrans}$, by communicating back and forth over a noisy memoryless channel $P_{Y|X}$. 
Specifically, we restrict our discussion to channels with a bi\textit{}nary input alphabet $\mathcal{X} = \{0,1\}$, and a general (possibly continuous) output alphabet $\mathcal{Y}$. Note that while the order of speakers in the input protocol itself might be determined on the fly (by the sequence of functions $\bs{\psi}$), we restrict the simulating protocol to use a predetermined order of speakers, due to the fact that our physical channel model does not allow simultaneous transmissions
(this point is elaborated in Section~\ref{section:previous}).

To achieve their goal, Alice and Bob employ a length-$\sigmalength$ coding scheme $\simulatingprotocol$ that uses the channel $\sigmalength$ times. The coding scheme consists of a disjoint partition $\tilde{A}\subseteq\{1,...,\sigmalength\}$, 
$\tilde{B} =\{1,...,\sigmalength\}\setminus \tilde{A}$,
where $\tilde{A}$ (resp. $\tilde{B}$) is the set of time indices where Alice (resp. Bob) speaks. This disjoint partition can be a function of $\bs{\psi}$, but not of $\bs{\phi}^\text{Alice},\bs{\phi}^\text{Bob}$. At time $j\in \tilde{A}$ (resp. $j\in \tilde{B}$), Alice (resp. Bob) sends some \textit{deterministic} function $X_j$ of $(\bs{\phi}^\text{Alice}, \bs{\psi}$) (resp. $(\bs{\phi}^\text{Bob}, \bs{\psi}$)), and of everything she has received so far from her counterpart. The transmitted $X_j$ is observed by Bob (resp. Alice) through the channel $P_{Y|X}$, whose output is denoted by $Y_j$. Note that we assume that $Y_j - X_j - (X^{j-1}, Y^{j-1})$ forms a Markov chain. The rate of the scheme is $R = \frac{\protLength}{\sigmalength}$ bits per channel use. When communication terminates, Alice and Bob produce their \textit{simulations} of the transcript $\protTrans$, denoted  by $\hat{\bs{\protTrans}}_A(\simulatingprotocol,\bs{\phi}^\text{Alice}, \bs{\psi})\in\{0,1\}^\protLength$ and  $\hat{\bs{\protTrans}}_B(\simulatingprotocol,\bs{\phi}^\text{Bob}, \bs{\psi})\in\{0,1\}^\protLength$ respectively. The error probability attained by the coding scheme is the probability that either of these simulations is incorrect, i.e., 
\begin{align}\label{eq:proterror}
&\Pe(\simulatingprotocol,\bs{\transFunc})\dfn
\\
&\Pr\left(
\hat{\bs{\protTrans}}_A(\simulatingprotocol,\bs{\phi}^\text{Alice}, \bs{\psi})\neq \bs{\protTrans} \;\vee\; \hat{\bs{\protTrans}}_B(\simulatingprotocol,\bs{\phi}^\text{Bob}, \bs{\psi})
\neq \bs{\protTrans}\right).
\end{align}
A rate $R$ is called \textit{achievable} if there exists a sequence $\simulatingprotocol_{\protLength}$ of length-$N_n$ coding schemes that operate on length-$\protLength$ input protocols $\bs{\transFunc}$, where $\frac{\protLength}{\sigmalength_{\protLength}}\geq R$, and attain a vanishing worst-case error probability, i.e., 
\begin{align}
\lim_{\protLength\to \infty}\max_{\text{protocols } \bs{\transFunc} \text{ of length } \protLength} P_e(\simulatingprotocol_{\protLength},\bs{\transFunc})  = 0.
\end{align}
Accordingly, we define the \textit{essential interactive capacity}  $\Cinter(P_{Y|X})$ as the supremum of all achievable rates for the channel $P_{Y|X}$. \newtextrevb{This definition is more conservative than all other interactive capacity definitions appearing in the literature, as further discussed in Section~\ref{section:previous}. In particular, note that our capacity definition makes worst case assumptions on the input protocol, and is hence a function of the channel model only. We also note in passing that our assumptions on channel access model are conservative and not worst case, as we permit any predetermined scheduling of speakers (more on that below). This approach makes sense practically, since there seems to be no fundamental reason to limit Alice and Bob in terms of which coding scheme they can use. Moreover, taking a worst case approach in terms of channel access can lead to trivialities, since there exist pessimistic access schedules (e.g., allocating only a single channel use for Alice) that would render the capacity zero.} 


For simplicity of exposition, we restrict our discussion from this point on to binary-input channels. Since at least $\protLength$ bits need to be exchanged in order to reliably simulate a general length-$\protLength$ input protocol, the essential interactive capacity for such channels must satisfy  $\Cinter(P_{Y|X})\leq 1$. In the special case of a noiseless channel, i.e., where the output deterministically reveals the input bit, and assuming that the order of speakers is predetermined (namely $\bs{\psi}$ contains only constant functions), this upper bound can be trivially achieved; Alice and Bob can simply evaluate and send $\protTrans_{\tind}$ sequentially according to \eqref{eq:transfunc} and \eqref{eq:transfunc2}. Note however, that if the order of speakers is general, then this is not a valid solution, since we required the order of speakers in the coding scheme to be fixed in advance. Nevertheless, any general length-$\protLength$ input protocol can be sequentially simulated using the channel $2\protLength$ times with alternating order of speakers, where each party sends a dummy bit whenever it is not their time to speak. Conversely, a factor two blow-up in the input protocol length in order to account for a non pre-determined order of speakers is also necessary. To see this, consider an example of an input protocol where Alice's first bit determines the identity of the speaker for the rest of time; in order to simulate this protocol using a predetermined order of speakers, it is easy to see that at least $n-1$ channel uses must be allocated to each party in advance. We conclude that under our restrictive capacity definition, the essential interactive capacity of a noiseless (binary-input) channel is exactly $\frac{1}{2}$. \newtextrevb{It is instructive to note that for a noiseless channel, one could have permitted the order of speakers to be determined on-the-fly, avoiding the need to pre-allocate the channel and eliminating the factor $1/2$ penalty. However, the truly noiseless channel is a singular case, since for any arbitrarily small channel error probability, using an adaptive order of speakers would yield channel access collisions, which are not supported in our channel model. We further elaborated on this point in Section~\ref{section:previous}}.

When the channel is noisy, a tighter trivial upper bound holds:
\begin{align}\label{eq:shannon_bound}
\Cinter(P_{Y|X})\leq\frac{1}{2}\Cshannon(P_{Y|X}), 
\end{align}
where $\Cshannon(P_{Y|X})$ is the Shannon capacity of the channel. To see this, consider the same example given above, and note that each party must have sufficient time to reliably send $n-1$ bits over the noisy channel. Hence, the problem reduces to a pair of one-way communication problems, in which the Shannon capacity is the fundamental limit. We remark that it is reasonable to expect the bound~\eqref{eq:shannon_bound} to be loose, since general input protocols cannot be trivially reduced to one-way communication as the parties cannot generate their part of the transcript without any interaction. However, the tightness of the bound remains a wide open question. We note in passing that if we had considered simulating only protocols with a predetermined order of speakers, the corresponding upper bound would have been $\Cinter(P_{Y|X})\leq\Cshannon(P_{Y|X})$.

\begin{remark}\newtextrevb{[The notion of determinism in interactive coding schemes]
Let us briefly discuss the difference between deterministic and randomized coding schemes for interactive communication. A deterministic coding scheme is one where the transmission functions used by Alice and Bob to generate their next channel inputs are fixed and given in advance; in other words, the channel inputs generated by both parties are solely determined by the input protocol and the channel outputs. A randomized coding scheme, on the other hand, is allowed to use random bits from an exogenous source; Namely, Alice / Bob pick a random function to apply to their data (which includes all their past observations) each time, and this function can be different even if the data it is applied to is the same. 
}

\newtextrevb{
We note that in principle, when working over stochastic memoryless channels, any randomized scheme can be converted into a deterministic one by extracting the needed random bits from the noisy channel outputs (e.g., using \cite{elias1972efficient,vonNeauman1951}). However, this procedure incurs a loss in rate due to the overhead of randomness extraction and possibly communication of randomness. While semantically, such a scheme might appear to be randomized, we note that it is in fact deterministic, since all the transmission functions used by the parties (including the ones used for randomness extraction) are fixed in advance. In a related context, see for example~\cite{leighton1986estimating}, where the authors construct an optimal randomized finite-state machine to estimate the bias of a coin, and then derandomize it by extracting the necessary random bits from the observations themselves, with a modest penalty in performance. We further observe that one could potentially define a more stringent notion of deterministic coding schemes, where the parties' inputs are not allowed to depend on the random channel outputs. However, while this definition would disallow any randomness extraction, it would also remove the interactive component from the problem. 
}
\end{remark}

\subsection{\newtext{Channel Models}}
The first noisy channel model we consider is the memoryless binary symmetric channel with crossover probability $0\leq\BSCprob\leq\frac{1}{2}$, BSC($\BSCprob$). The input to output relation of the BSC($\BSCprob$) is given by
\begin{align}
Y=X\oplus Z
\end{align}
where $X,Y,Z\in\binarySpace$, $\oplus$ denotes addition over $\binarySpace$. $Z$ is statistically independent of $X$ with $\Pr(Z=1)=\BSCprob$. We denote its Shannon capacity by
\begin{align}
\Cshannon(\BSCprob)\dfn 1-h(\BSCprob),
\end{align}
where $h(\BSCprob)\dfn -\BSCprob\log\BSCprob-(1-\BSCprob)\log(1-\BSCprob)$ is the binary entropy function, and $\log(x)\dfn\log_2(x)$. We {also} use $\Cinter(\BSCprob)$ to denote the essential interactive capacity of the BSC($\BSCprob$).

A richer channel model which is commonly used in the coding literature is the binary memoryless symmetric (BMS) channel \cite{gallagerPhD, GallagerIT, richardson2008modern,arikan08, i2012extremes}. While several equivalent definitions exist, the following definition of a BMS channel as a collection of BSC with various crossover probabilities \cite{pedarsani2011construction}, is most convenient for the derivations in this paper:
\begin{definition}\label{def:BMS}[BMS channels]
	A memoryless channel with binary input $X$ output $Y$ and a conditional distributions $P_{Y|X}$ is called \emph{binary memoryless symmetric} channel (BMS($P_{Y|X}$)) if there exists a sufficient statistic of $Y$ for $X$: $g(Y)=(X\oplus Z_T,T)$ , where $(T,Z_T)$ are statistically independent of $X$, $Z_T$ is a binary random variable with $\Pr(Z_T=1|T=t)=t$, and $0\leq T\leq\frac{1}{2}$ with probability one.
\end{definition}
Consequently, the Shannon capacity of BMS($P_{Y|X}$) channel is 
\begin{align}
\Cshannon(P_{Y|X})=1-\Expt h(T). 
\end{align}
\newtext{
The simplest example for a BMS channel is the BSC($\BSCprob$) for which $T=\BSCprob$ with probability one. 
The binary erasure channel with erasure probability $\epsilon$, BEC($\epsilon$), can be cast as a BMS channel taking  $T=\frac{1}{2}$ with probability $\epsilon$ and $T=0$ with probability $1-\epsilon$. It is in place to note, however, that in an actual BEC, a Bernoulli($1/2$) bit is not produced when $T=1/2$ (this subtle point is further discussed in Subsection~\ref{subsection:ties}).
The binary additive white Gaussian noise (BiAWGN) channel, $Y=X+Z$ where $X\in\{-1,+1\}$ and $Z\sim\m{N}(0,\sigma^2)$ is statistically independent of $X$, is also a BMS where $T$ is a continuous random variable on $[0,1/2]$ (see \cite[Chapter~4]{richardson2008modern}).
}

\section{\newtext{Main Contribution} \label{sec:maincontrib}}
\newtext{
In this paper, we derive a lower bound on the essential interactive capacity of any BMS channel, which depends only on its Shannon capacity. In particular, we show that the essential interactive capacity is always at least $0.0302$ of the Shannon capacity, uniformly for all BMS channels. Indeed, since $\Cinter(P_{Y|X}) \leq \frac{1}{2}\Cshannon(P_{Y|X})$ always holds (and is tight for noiseless BMS), then using the Shannon capacity as a yardstick is intuitively appealing, and our lower bound can be interpreted as saying that the ``cost of interactiveness'' is not too large. 
}

\newtext{
\begin{theorem}\label{theorem:interactiveboundBMS} 
	For any BMS($P_{Y|X}$) channel with positive Shannon capacity $\Cshannon(P_{Y|X})$ and essential interactive capacity $\Cinter(P_{Y|X})$
	\begin{align}
	\frac{\Cinter(P_{Y|X})}{\Cshannon(P_{Y|X})}\geq 0.0302.
	\end{align}
\end{theorem}
}

\newtext{
Note that Theorem~\ref{theorem:interactiveboundBMS} also applies to the special case of the BSC with any crossover probability. In fact, we first prove Theorem~\ref{theorem:interactiveboundBMS} for the BSC case, and them extend the result to general BMS channels.
}

\newtext{
The first step in the proof is standardly symmetrizing the order of speakers in the input protocol by possibly adding dummy transmissions, such that Alice speaks at odd times, and Bob speaks at even times, namely resulting in a modified protocol where $\bs{\psi} = \{\text{Alice, Bob, Alice, Bob, \ldots}\}$. In the sequel, we refer to this order of speakers as \textit{bit-vs.-bit}.
This reduces the rate by a factor of two at most. We then use a \textit{rewind-if-error} scheme in the spirit of \cite{schulman1992communication,kol2013interactive}, designed for simulating {the transcript of protocols} with bit-vs.-bit order of speakers. As explained in the chess game example, the {transcript} bits of an interactive {protocol} should in general be decoded instantaneously, which implies that error correction codes (that typically use long blocks) cannot be straightforwardly used. Instead, rewind-if-error schemes are based on \textit{uncoded transmission}, followed by error detection and retransmission. Namely, the {transcript} is simulated in blocks, as if no errors are present. 
{Then, an error detection phase takes place, initiating the retransmission of the block whenever errors are detected.}  The scheme presented in Sections~\ref{section:SchemeDescription} and \ref{section:schemeanalysis} of this paper is based on \textit{layered} error detection and retransmission. 
The rate of the proposed scheme is shown to be mostly effected by the efficiency of the error detection in the first layer. Thus we use extended-Hamming codes for error detection at that layer only, and a standard randomized error detection \cite{communicationComplexityBook} at higher layers. 
}

\newtext{
Our scheme is premised on the assumption that the channel is unlikely to introduce any errors within a single block. If this is not the case, we first standardly apply repetition coding in order to reduce the error rate to the desired level; crucially, we show that a sufficient number of repetitions in the BMS case is inversely proportional to the Shannon capacity of the channel. Accounting for this repetition overhead, calculating the rate of the rewind-if-error coding scheme, and judicially tuning its parameters, we show that this scheme yields the lower bound of Theorem~\ref{theorem:interactiveboundBMS}. 
}


\newtext{
Finally, while the scheme delineated above is randomized, we show that it can be converted to a fully deterministic scheme with an asymptotically vanishing rate loss, which makes our bound applicable in the essential interactive capacity setting. To that end, using a careful concentration analysis appearing in Appendix~\ref{appendix:concentration}, we first show that the number of random bits required by our scheme is only $o(n)$. Then, we harvest these bits from the channel via standard techniques, using only $o(n)$ channel uses. This process, which makes our scheme completely deterministic, clearly has a negligible effect on its overall rate. 
}

\newtext{
\section{Connections to the Existing Work \label{section:previous}}
In this section, we put our definition of essential interactive capacity in context of the existing literature. While the classical Shannon capacity of a one-way channel has single agreed-upon definition that depends on the channel model $P_{Y|X}$ only, the same is not true in the interactive setting, where various distinct notions of capacity exist, drastically depending on different possible assumptions. Let us review these assumptions, and point out that our definition is always on the more restrictive side. 
}

{
\begin{itemize}
\item  \textit{Order of speakers}. One can assume that the input protocol, $\bs{\transFunc}$, has either a predetermined order of speakers, or a general (adaptive) one. For a predetermined order, one can further assume that it has some fixed period (e.g.,  bit-vs.-bit). Our capacity definition does not restrict the order of speakers, hence our lower bound applies in any such setting (and for example, would increase by a factor of two if the order of speakers is bit-v.s-bit).
\item  \textit{Randomness resources}. In their coding scheme (simulating protocol), Alice and Bob can be allowed to use some exogenous source of (public or private)  randomness, in which case the scheme is called \textit{randomized}, or are not allowed to use any exogenous randomness, in which case the scheme is called \textit{deterministic}. We emphasize that in the deterministic case, the channel inputs are uniquely determined by the input protocol and the noisy channel output sequences. Our capacity definition makes the more stringent assumption of allowing only deterministic schemes, hence our lower bound applies to all cases. 
\item \textit{Rate definition}: The coding scheme can be either fixed-length or variable-length. In the fixed-length case, the number of allocated channel uses is determined in advance, and the rate is simply the ratio between the protocol length and the number of channel uses. In the variable-length case, the length of the protocol or the number of channel uses is allowed to be random, and the rate is then the ratio between the expected protocol length and the expected number of channel uses (\newtextrevb{though worst case length analysis also appears in the literature, for example \cite{agrawal2016adaptive,dani2018interactive}}). Our capacity definition adopts the more stringent fixed-length setting, hence our lower bound applies to all cases. 

\item \textit{Physical channel model}: There are two distinct assumptions that can be made on the underlying structure of the channel. In one setting \cite{schulman1992communication,schulman1996coding,kol2013interactive}, Alice and Bob are not allowed (at the physical level) to simultaneously access the channel; they must decide in advance who uses the channel at each time point. In another (richer) setting \cite{ghaffari2014optimal,haeupler2014interactive}, Alice and Bob communicate over a general two-way channel~\cite{shannon1961two}, which means that they both input a symbol to the channel at any given time. In the interactive communication literature, a certain two-way channel has received attention. In this model, Alice and Bob each have \textit{three} input symbols $\{0,1, \mathsf{silence}\}$, and binary output symbols. A party that is not silent receives a zero. If one party is silent and the other is not, the silent one sees the input of its counterpart via a BSC. If both are silent, they observe uniform independent noise\footnote{This has in fact been considered in the adversarial setting, where in the case that both parties are silent, it was assumed that they observe undetermined symbols. What we described above is arguably the most natural way to adapt this adversarial assumption to the probabilistic setting.}. Our capacity definition adopts the more basic setting where no simultaneous channel access is allowed; since any two-way channel can be used this way, our lower bounds essentially applies to all cases. 
\item \textit{Input protocol}: In the interactive communication literature, it is commonly assumed that the redundancy of the coding scheme is measured with respect to the communication complexity of a function, and the interactive capacity corresponds to the worst case blow-up over all functions (as further explained below). Alternatively, as suggested in this paper, one can measure the redundancy of the coding scheme with respect to \textit{any protocol} (unrelated to any optimal function computation problem), and then the (essential) interactive capacity corresponds to the worst case blow-up over all protocols. Since our capacity definition normalizes by the length of the input protocol, it is at least in principle stricter than the one using communication complexity, and hence our lower bounds apply in both cases. 
\end{itemize}
}

\newtext{
Let us now review the main relevant literature. The interactive communication problem introduced by Schulman \cite{schulman1992communication, schulman1996coding} is motivated by Yao's communication complexity scenario \cite{yao1979some}. In that latter scenario, the input of a function $f$ is distributed between Alice and Bob, who wish to compute $f$ with negligible error by exchanging (noiseless) bits using some interactive protocol. The length of the shortest protocol achieving this is called the \textit{communication complexity} of $f$, and denoted by $CC(f)$. In Schulman's (random) interactive communication setup, Alice and Bob must achieve their goal by communicating through a pair of independent noisy channels, where the physical model does not allow simultaneous transmissions. For that setup, Schulman showed that one can attain this goal with negligible error, using only a constant blow-up in the length of the communication. 
}

\newtext{
In \cite{kol2013interactive}, Kol and Raz considered the interactive communication problem, with no simultaneous transmissions, over a BSC($\BSCprob$). They denoted the minimal \textit{expected} length of a coding scheme computing $f$ with a negligible error probability, by $CC_{\BSCprob}(f)$. They then defined the corresponding interactive capacity as:
\begin{align}\label{eq:kolrazCi}
\Cinter^{\mathsf{KR}}(\BSCprob) \triangleq \lim_{\protLength\to\infty}\min_{f:CC(f)=n}\frac{\protLength}{CC_{\BSCprob}(f)}. 
\end{align}
with the additional assumption that the order of speakers in the input protocol is predetermined. They proved that 
\begin{align}\label{eq:kolrazrate_upper}
\Cinter^{\mathsf{KR}}(\BSCprob) \leq  1-\Omega\left(\sqrt{h(\BSCprob)}\right). 
\end{align}
in the limit of $\BSCprob\to 0$. They further proved that a rate of $1-O(\sqrt{h(\BSCprob)})$ is achievable under an additional assumption that the order of speakers in the input protocol is has a small period. The assumption on the order of speakers is crucial. Indeed, consider again the example where the function $f$ is either Alice's input or Bob's input as decided by Alice. In this case, the communication complexity with a predetermined order of speakers is double that without this restriction, and hence considering such protocols renders $\Cinter^{\mathsf{KR}}(\BSCprob) \leq \frac{1}{2}$. For further discussion on speaking order impact as well as channel models that allow collisions, see \cite{haeupler2014interactive}. Note that our definition of the BSC essential interactive capacity is stricter than~\eqref{eq:kolrazCi}, at least in principle, both since the latter does not consider adaptive input protocols, and also since we measure our blow-up w.r.t. the length of the entire transcript. For this reason, $\Cinter(\BSCprob) \leq \Cinter^{\mathsf{KR}}(\BSCprob)$, hence our lower bound applies to $\Cinter^{\mathsf{KR}}(\BSCprob)$ as well (and also achieves the asymptotic behavior~\eqref{eq:kolrazCi} when simulating bit-vs.-bit protocols). Our capacity definition further enjoys the property of being decoupled from any source coding problem such as function communication complexity.
}

\newtext{
For a fixed nonzero $\BSCprob$, the coding scheme presented in \cite{schulman1992communication} (which precedes \cite{kol2013interactive}) implies that $\Cinter(\BSCprob) \geq \alpha \cdot \Cshannon(\BSCprob)$ for some universal constant $\alpha$, but the constant has not been computed (and to the best of our knowledge, has not been computed for any scheme hitherto). Both \cite{schulman1992communication} and \cite{kol2013interactive} based their proofs on rewind-if-error coding schemes, i.e., schemes based on a \textit{hierarchical} and \textit{layered} error detection and appropriate retransmissions, which is also the approach we take in this paper. 
}

\newtext{
In~\cite{haeupler2014interactive}, Haeupler considered a different physical channel model where Alice and Bob can access the channel simultaneously and have three input symbols (as essentially described above). In this setup, he showed that a rate of $1-O(\sqrt{\BSCprob})$ is achievable for any alternating input protocol, which is higher than the upper bound~\eqref{eq:kolrazrate_upper}. His results also hold in the more difficult adversrial setting assuming shared randomness, and reduces slightly to $1-O\left(\sqrt{\BSCprob} \log\log\frac{1}{\BSCprob}\right)$ when no randomness is available. 
}

\newtext{
Let us now discuss the issue of randomness resources. The scheme in \cite{schulman1992communication} requires only private randomness, while \cite{kol2013interactive} requires public randomness. It is interesting to note that Schulman's \textit{tree code} scheme \cite{schulman1996coding} is not randomized. However, it is not designed to be rate-efficient, and for example does not achieve the lower bound in~\cite{kol2013interactive}. A non-random coding scheme was recently proposed by Gelles et. al. \cite{gelles2018explicit} based on a concatenation of a derandomized interactive coding scheme and a tree-code. This scheme achieves a rate $1-O(\sqrt{h(\BSCprob)})$
which is also the rate of the rewind-if-error scheme in this paper in the limite of $\BSCprob\to 0$ as stated in Corrolary~\ref{corollary:Obehaviour}.
The rewind-if-error {scheme} presented in this paper is inspired by the scheme in \cite{kol2013interactive}, yet its error detection mechanism is not based on random hashes, but rather on extended-Hamming codes and randomized (yet structured) error detection. The deterministic coding scheme presented in Section~\ref{section:derandomize} is not based on derandomization as in \cite{gelles2018explicit}, but rather on suitably adapting the error detection and using concentration analysis to show that it requires only small number of random bits. These bits are then extracted from the noisy channels in a standard way using a small number of channel uses, which are taken into account in the overall rate calculation. We emphasize that our coding scheme is fully deterministic, namely, the channel inputs generated by Alice and Bob are uniquely determined by the input protocol and the channel noise sequences only. 
}

\newtext{
Other channel models have been addressed in the literature. Much work has been dedicated to the adversarial setting, where the channel is controlled by an adversary with some limited jamming budget, see for example \cite{schulman1996coding,ghaffari2014optimal,haeupler2014interactive,agrawal2016adaptive}. It is important to note that the rewind-if-error approach and the randomness extraction ideas we use, do not apply in adversarial settings. More recently, interactive communication over channels with noiseless feedback has been studied in \cite{gelles2017capacity}.
}

\newtext{
To summarize the discussion above, there are various distinct setups and sets of assumptions one may wish to consider when studying interactive communication, which can have significant effect on the fundamental limits. Our definition of capacity, and its corresponding lower bound, are based on the most restrictive set of assumptions: the order of speakers in the input protocol can be adaptive, but is predetermined in the coding scheme; the coding is fixed-length and the blow-up is computed relative to the length of the input protocol; and no private or public randomness are allowed. Consequently, our capacity lower bounds remain valid for any other set of assumptions.
}

\newtext{
Finally, we note in passing that the current study extends our preliminary results presented in \cite{InterLowerBoundISIT} in the following aspects: 
i) The error detection in the scheme is structured and is not based on random hashes. ii) The rate of the resulting scheme is improved and consequently the lower bound for the ratio between the essential interactive capacity and the Shannon capacity is also improved. iii) The scheme described in this paper deterministic. iv) The results are generalized from the BSC to arbitrary BMS channels.
}

\section{Preliminaries\label{section:preliminaries}}
Let $D(P||Q)\dfn \sum_{x\in\mathcal{X}}P(x)\log\frac{P(x)}{Q(x)}$ denote the Kullback-Leibler Divergence between the distributions $P(\cdot)$ and $Q(\cdot)$.
Let $d(p||q)\dfn p\log\frac{p}{q}+(1-p)\log\frac{1-p}{1-q}$ denote the Kullback-Leibler Divergence between two Bernoulli random variables with probabilities $p$ and $q$. 
In the sequel we use $\indfunc{\cdot}$ to denote the indicator function, which equals one if the condition is satisfied and zero otherwise. 

The following simple results are used throughout the paper:
\begin{lemma}[Repetition coding over BSC] \label{lemma:repetition} Let a bit be sent over  BSC($\BSCprob$) using $\rho$  repetitions and decoded by {a} majority vote (if $\rho$ is even, ties are broken by tossing a fair coin). 
The decoding error probability $\Pe$ can be upper bounded by
	\begin{align}\label{eq:repchernoff}
	\Pe\leq\beta^\rho	=2^{-\rho \cdot d(\frac{1}{2}||\BSCprob)}, 
	\end{align}	
where $\beta\dfn 2\sqrt{\BSCprob(1-\BSCprob)}$ is the Bhattacharyya parameter respective to the BSC($\BSCprob$). The induced channel from the input bit to its decoded value is thus a BSC$(\Pe)$.
\end{lemma}
The proof is standard (see for example \cite{GallagerIT}) and can be regarded as special case of Lemma~\ref{lemma:BMSrepetition} stated and proved in Section~\ref{section:bms}. Note that the random tie breaking is done in order to simplify the scheme and its analysis. It does, however, assume private randomness at both parties. In Section~\ref{section:derandomize} we show how the random tie breaking can be circumvented.

We now introduce two error detection methods that would be used in the coding scheme. The first one assumes the error are generated by BSC's and is based on error correction codes:
\begin{definition}[Error detection using an extended-Hamming code]\label{def:exhammingdetection}
Let $\undcodedSim^A$ and $\undcodedSim^B$ be binary (row) vectors of length $\levelblock$ held by Alice and Bob respectively. 
Let $H$ be the parity check matrix of an {extended-Hamming code} with parameters $(\levelblock,\levelblock-\log\levelblock-1,4)$. 
Let $NEQ$ be a variable set to one if the parties decide that $\undcodedSim^A \neq \undcodedSim^B$ and set to zero otherwise,  calculated according to the following algorithm:
\begin{enumerate}
	\item Alice calculates her syndrome vector $\syndromevec^A =  \undcodedSim^A H^T$
	\item Bob calculates his syndrome vector $\syndromevec^B=\undcodedSim^B H^T$
	\item Alice sends $\syndromevec^A$ ($1+\log\levelblock$ bits) to Bob
	\item Bob {calculates} $NEQ=\indfunc{\syndromevec^A\neq\syndromevec^B}$
	\item Bob sends $NEQ$ ($1$ bit) to Alice	
\end{enumerate}
The overall number of bits communicated between Alice and Bob is $2+\log\levelblock$. 
\end{definition}

The performance of this scheme over a BSC($\BSCprob$) is given in the following lemma:
\begin{lemma}
Assume that
\begin{align}
\undcodedSim^A = \undcodedSim^B\oplus\noisevec.
\end{align}	
where $\noisevec$ is an i.i.d $\text{Bernoulli}(\BSCprob)$ vector.
The probability of a mis-detected error of the scheme in Definition~\ref{def:exhammingdetection} is given by
\begin{align}\label{eq:hammingmisdetect}
&\Pr\left(NEQ=0, \undcodedSim^A\neq\undcodedSim^B\right)\\
&=\frac{1}{2\levelblock}\left(1+2(\levelblock-1)(1-2\BSCprob)^{\frac{\levelblock}{2}}+(1-2\BSCprob)^\levelblock\right)\\
&\quad-(1-\BSCprob)^\levelblock.
\end{align}
The corresponding probability of a false error detection is
\begin{align}
\Pr\left(NEQ=1, \undcodedSim^A=\undcodedSim^B\right)=0.
\end{align}
\end{lemma}

\begin{proof}	
First, it is clear that for any $\undcodedSim^A = \undcodedSim^B$ we have $NEQ=\indfunc{\syndromevec^A\neq\syndromevec^B}=0$ with probability one, so 
the probability of false error detection is $\Pr\left(NEQ=1, \undcodedSim^A=\undcodedSim^B\right)=0$. For the probability of error mis-detection, note that $\syndromevec^A\oplus\syndromevec^B =  (\undcodedSim^A\oplus\undcodedSim^B) H^T=\noisevec H^T$. Therefore, the event $NEQ=0$ is identical to the event in which $\syndromevec^A\oplus\syndromevec^B=\noisevec H^T =\boldsymbol{0}^T$, i.e., $\noisevec$ is a codeword in $H$. All in all
\begin{align}
&\Pr\left(NEQ=0, \undcodedSim^A\neq\undcodedSim^B\right)\\
&=\Pr\left(\noisevec H^T =\boldsymbol{0}^T, \noisevec\neq \boldsymbol{0}^T \right)\\
&=\frac{1}{2\levelblock}\left(1+2(\levelblock-1)(1-2\BSCprob)^{\frac{\levelblock}{2}}+(1-2\BSCprob)^\levelblock\right)\\
&\quad -(1-\BSCprob)^\levelblock,\label{eq:misdetectHamming}
\end{align}
where \eqref{eq:misdetectHamming} is standardly calculated using the dual code \cite[p. 52]{klove2012error}.
\end{proof}

The second error detection scheme is a randomized scheme based on \cite[p. 30]{communicationComplexityBook}, which applies for arbitrary vectors. We note that this scheme performs the error detection using hashing, where the hash functions are implemented using polynonmials.
\begin{definition}[Randomized error detection using polynomials]\label{def:equalityPrivate}
Let $\undcodedSim^A$ and $\undcodedSim^B$ be arbitrary binary vectors of length $\ell$ held by Alice and Bob respectively. Let $\Pmd\in\mathbb{N}$, where $\Pmd>1$. Let $q$ be a prime number such that $\Pmd \ell \leq q\leq 2\Pmd \ell$ (by Bertrand's postulate such a number must exist). 
Let $NEQ^{Poly}$ be a variable set to one if the parties decides that $\undcodedSim^A \neq \undcodedSim^B$ and set to zero otherwise, calculated according to the following algorithm:
\begin{enumerate}
	\item Alice uniformly draws $U\in \GFq$, $U\neq 0$ .
	\item Alice calculates $A(U,\undcodedSim^A)=\sum_{i=1}^\ell X_i^A U^{i-1} (\text{mod } q)$
	\item Alice sends Bob $U$ and $A(U,\undcodedSim^A)$
	\item Bob calculates $B(U,\undcodedSim^A)=\sum_{i=1}^\ell X_i^B U^{i-1} (\text{mod } q)$
	\item Bob {calculates} $NEQ^{Poly}=\indfunc{A(U,\undcodedSim^A)-B(U,\undcodedSim^A)\neq 0}$ 
	\item Bob sends $NEQ^{Poly}$ to Alice	
\end{enumerate} 
All in all, Alice needs to send at most $\lceil\log 2\Pmd\ell\rceil$ bits for the representation of 
$U$, and at most $\lceil\log 2\Pmd\ell\rceil$ bits for the representation of $A({U},\undcodedSim^A)$. Bob sends Alice one bit.
\end{definition}
\begin{lemma}\label{lemma:poly}
The error detection scheme of Definition~\ref{def:equalityPrivate} obtains an error mis-detection probability of 
\begin{align}
\Pr\left(NEQ^{Poly}=0\mid \undcodedSim^A\neq\undcodedSim^B\right)\leq \frac{1}{\Pmd},
\end{align}
and a false error detection probability of 
\begin{align}
\Pr\left(NEQ^{Poly}=1\mid \undcodedSim^A=\undcodedSim^B\right)=0.
\end{align}
\end{lemma}

\begin{proof}
Note that $A(U,\undcodedSim^A)$ and $B(U,\undcodedSim^B)$ are the evaluation at point $U$ of two polynomials over $\GFq$ whose (binary) coefficients are the elements of $\undcodedSim^A$ and $\undcodedSim^B$ respectively. Clearly, if $\undcodedSim^A=\undcodedSim^B$, then $NEQ=0$ for every value of $U$ hence $\Pr\left(NEQ=1\mid \undcodedSim^A=\undcodedSim^B\right)=0$.
On the other hand, if $\undcodedSim^A\neq\undcodedSim^B$, $A(U,\undcodedSim^A)-B(U,\undcodedSim^A)=0$ implies that $U$ is a root of the polynomial 
\begin{align}\label{eq:checkpoly}
\sum_{i=1}^\ell (X_i^A-X_i^B) U^{i-1} (\text{mod } q).
\end{align}
Since the degree of the polynomial is at most $\ell$, there are at most $\ell-1$ such roots, so 
\begin{align}
\Pr\left(NEQ=0
\mid \undcodedSim^A\neq\undcodedSim^B\right)\leq \frac{\ell-1}{q}< 
\frac{\ell}{\Pmd\ell}=\frac{1}{\Pmd}.
\end{align}

\end{proof}


\newtext{\section{The Lower Bound in the BSC Case\label{section:mainresults}}
In the following sections we prove the lower bound on the essential interactive capacity to the BSC case, which is then extended to BMS in Section~\ref{section:bms}. The BSC version of the bound is stated in the following theorem:
\begin{theorem}\label{theorem:interactivebound} 
For any BSC with crossover probability $0\leq\BSCprob	\leq {1}/{2}$, Shannon capacity $\Cshannon(\BSCprob)$ the and essential interactive capacity $\Cinter(\BSCprob)$ the following bound holds:
	\begin{align}
	\frac{\Cinter(\BSCprob)}{\Cshannon(\BSCprob)}\geq 0.0302.
	\end{align}
\end{theorem}
This bound is derived by using a rewind-if-error scheme 
for a small $\BSCprob$, whose rate appears in Theorem~\ref{theorem:rewindifrateHamming}, and then leveraging it to a general BSC using repetition coding via Lemma~\ref{lemma:gammalemma}.
}
\begin{theorem}\label{theorem:rewindifrateHamming}
The transcript of any protocol with n bit-vs.-bit order of speakers (i.e. Alice sends a bit on odd times and Bob sends a bit on even times), can be reliably simulated over BSC($\BSCprob$) (i.e. with a vanishing error as $\protLength\to\infty$ for a fixed $\BSCprob$)
at the rate specified in \eqref{eq:hammingrate},
\begin{figure*}[!hb]
\rule{\textwidth}{1pt}
\begin{align}\label{eq:hammingrate}
\RBsc(\BSCprob,\levelblock)\dfn
\frac
{
	1- \levelblock\BSCprob-(3+\log\levelblock)\batdistance^{5}
	-\frac{\levelblock^2}{\levelblock-1}
	\left(
	\Pfailure_{1}
	+3\batdistance^{7}\levelblock\log\levelblock
	\frac{2-\batdistance^2\levelblock}{(1-\batdistance^2\levelblock)^2}
	\right)
	-
	3\batdistance^{7}\log\levelblock
	\frac{2-\batdistance^2}{(1-\batdistance^2)^2}
}
{
	1+\frac{5(3+\log\levelblock)}{\levelblock}
	+3\log\levelblock\left[
	\frac{3(2\levelblock-1)}{k(\levelblock-1)^2}
	+\frac{4\levelblock}{(\levelblock-1)^3}
	+\frac{4\levelblock-2}{\levelblock(\levelblock-1)^2}
	\right]
}
\end{align}
\end{figure*}
where
\begin{align}
\Pfailure_{1}&\leq 
\frac{1}{2\levelblock}\left(1+2(\levelblock-1)(1-2\BSCprob)^{\frac{\levelblock}{2}}+(1-2\BSCprob)^\levelblock\right)\\
&-(1-\BSCprob)^\levelblock+(3+\log\levelblock)\batdistance^{5}.
\label{ea:Pe1}
\end{align} 
Let $0<\BSCprob<\frac{1}{16}$ and $\batdistance\dfn2\sqrt{\BSCprob(1-\BSCprob)}$.
$\levelblock$ can be any integer power of two satisfying $\levelblock\leq \frac{1}{8\BSCprob}$.
\end{theorem}
An example for $\RBsc(\BSCprob,\levelblock)$ with $\levelblock=9$ is depicted in Figure~\ref{fig:ratefigure}.
\begin{figure}
	\begin{center}
		\hspace*{-9mm}
		\includegraphics[width=1.2\columnwidth]{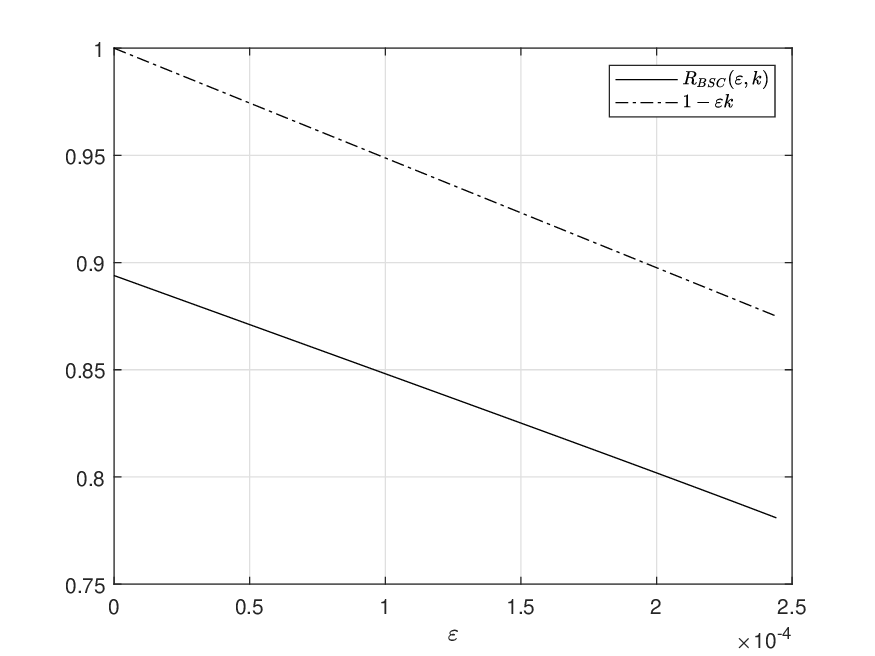}   
	\end{center}
	\caption{\newtext{
	An example for the rate $\RBsc(\BSCprob,\levelblock)$ from Theorem~\ref{theorem:rewindifrateHamming} with $\levelblock=2^9$. 
	The maximal channel crossover probability $\BSCprob$ is $\frac{1}{8\levelblock}$ as required by the theorem. It is observable that in this regime, $\RBsc(\BSCprob,\levelblock)$ is almost linear in $\BSCprob$. This is since the denominator is constant and the dominant term in the numerator is $1-\levelblock\BSCprob$. The function $1-\levelblock\BSCprob$ is also plotted as reference. 
	\label{fig:ratefigure}}}
\end{figure}
Using this theorem, $\Cinter(\BSCprob)\geq\max_{\levelblock}\RBsc(\BSCprob,\levelblock)$ for {protocols} with a bit-vs.-bit order of speakers and $\Cinter(\BSCprob)\geq\frac{1}{2}\max_{\levelblock}\RBsc(\BSCprob,\levelblock)$ for {protocols} with a general (possibly adaptive) order of speakers.

The proof of Theorem~\ref{theorem:rewindifrateHamming} is by the construction and analysis of a rewind-if-error scheme and appears in Sections~\ref{section:SchemeDescription} and \ref{section:schemeanalysis}. We note that the presented scheme is randomized and in Section~\ref{section:derandomize} we explain how to modify it to be deterministic. It is also in place to note that the error probability of this scheme decays polynomially in $\protLength$, as can be seen in the analysis of the error event.

The following corollary proved in Appendix~\ref{appendix:Oofh} states that the scheme obtains the rate lower bound  $1-O(\sqrt{h(\BSCprob)}$ from \cite{kol2013interactive}:
\begin{corollary}\label{corollary:Obehaviour}
For $\BSCprob\to 0$ 
		\begin{align}
		\max_{\levelblock} \RBsc(\BSCprob,\levelblock)\geq 1-O(\sqrt{h(\BSCprob)})
		\end{align}
\end{corollary}

As stated before, the presented rewind-if-error scheme is designed for BSC with a sufficiently small $\BSCprob$. 
For larger values of $\BSCprob$, the channel can be converted to a BSC($\delta'$) with $\delta'\leq\delta<\BSCprob$ using $\nrepetitions(\BSCprob,\targetProb)$ repetitions followed by a majority vote according to Lemma~\ref{lemma:repetition}.
The following lemma bounds the essential interactive capacity by using an interactive coding scheme augmented by a repetition code: 
\begin{lemma}\label{lemma:gammalemma}
	For every $0<\BSCprob<\frac{1}{2}$ and $0<\targetProb<\frac{1}{2}$ 
	\begin{align}
	\frac{\Cinter(\BSCprob)}{\Cshannon(\BSCprob)}\geq
	\frac{\Cinter(\delta)}{\log\frac{1}{\targetProb}+1}. 
	\end{align} 
\end{lemma}
\begin{proof}
Let $\nrepetitions$ be the smallest integer such that $\batdistance^\nrepetitions\leq \targetProb$, where 
$\batdistance\dfn2\sqrt{\BSCprob(1-\BSCprob)}$ is the Bhattacharyya parameter of the BSC($\BSCprob$) as above. By Lemma~\ref{lemma:repetition}, using $\nrepetitions$ repetitions, the BSC($\BSCprob$) can be converted to a BSC($\targetProb'$) with $\targetProb'\leq\targetProb$. Normalizing by $\Cshannon(\BSCprob)$ and noting that $\Cinter(\delta)\leq \Cinter(\delta')$ we obtain
\begin{align}\label{eq:CIbound}
\frac{\Cinter(\BSCprob)}{\Cshannon(\BSCprob)}\geq
\frac{\Cinter(\delta)}{\nrepetitions(\BSCprob,\targetProb)\Cshannon(\BSCprob)}. 
\end{align} 
By the definition of $\nrepetitions$ in Lemma~\ref{lemma:repetition}:
\begin{align}
\nrepetitions\leq \nrepetitions(\BSCprob,\targetProb)\dfn \frac{\log\frac{1}{\targetProb}}{\log\frac{1}{\batdistance}}+1,
\end{align}
where  `$+1$' accounts for rounding to the nearest larger integer. Furthermore,
\begin{align}\
\nrepetitions(\BSCprob,\targetProb)\Cshannon(\BSCprob)&=
\left(\frac{\log\frac{1}{\targetProb}}{\log\frac{1}{\batdistance}}+1\right)\Cshannon(\BSCprob)\label{eq:lemma4}\\
&\leq \frac{I(X;Y)}{L(X;Y)}\log\frac{1}{\targetProb}+I(X;Y),
\end{align}
where $X\sim\Ber\left(\tfrac{1}{2}\right)$ is the input of a BSC($\BSCprob$) channel and $Y$ is its respective output,
\begin{align}
I(X;Y)=D\left(P_{XY}||P_{X}P_Y\right)=\Cshannon(\BSCprob)
\end{align}
is the mutual information between $X$ and $Y$ and 
\begin{align}
L(X;Y)=D\left(P_{X}P_Y||P_{XY}\right)=d\left(\tfrac{1}{2}||\BSCprob\right)=\log\frac{1}{\batdistance}
\end{align}
is the \textit{lautum information} between $X$ and $Y$ \cite{palomar2008lautum}. Using the facts that
for the BSC, $L(X;Y)\geq I(X;Y)$ \cite[Theorem 12]{palomar2008lautum} and that trivially $I(X;Y)\leq 1$, concludes the proof.
\end{proof}

Theorem~\ref{theorem:interactivebound} now follows by using 
$\frac{1}{2}\RBsc(\targetProb,\levelblock)$ from Theorem~\ref{theorem:rewindifrateHamming} 
as a lower bound to $\Cinter(\targetProb)$, where the $\frac{1}{2}$ factor
is used in for the symmetrization of the order of speakers. Then, applying  Lemma~\ref{lemma:gammalemma}, gives:
\begin{align}\label{eq:ratokdelta}
\frac{\Cinter(\BSCprob)}{\Cshannon(\BSCprob)}\geq
\frac{\frac{1}{2}\RBsc(\targetProb,\levelblock)}{\log\frac{1}{\targetProb}+1}. 
\end{align} 
The bound in \eqref{eq:ratokdelta} is then tightened by scanning through various values of $\targetProb,\levelblock$ as seen in Figure~\ref{fig:capratio}. The combination of 
$\levelblock=2^9$ and $\targetProb=0.0001842$ gives the value of the lower bound in 
Theorem~\ref{theorem:interactivebound}.

\begin{figure}[h]
	\begin{center}
		\hspace*{-5mm}
		\includegraphics[width=1.15\columnwidth]	{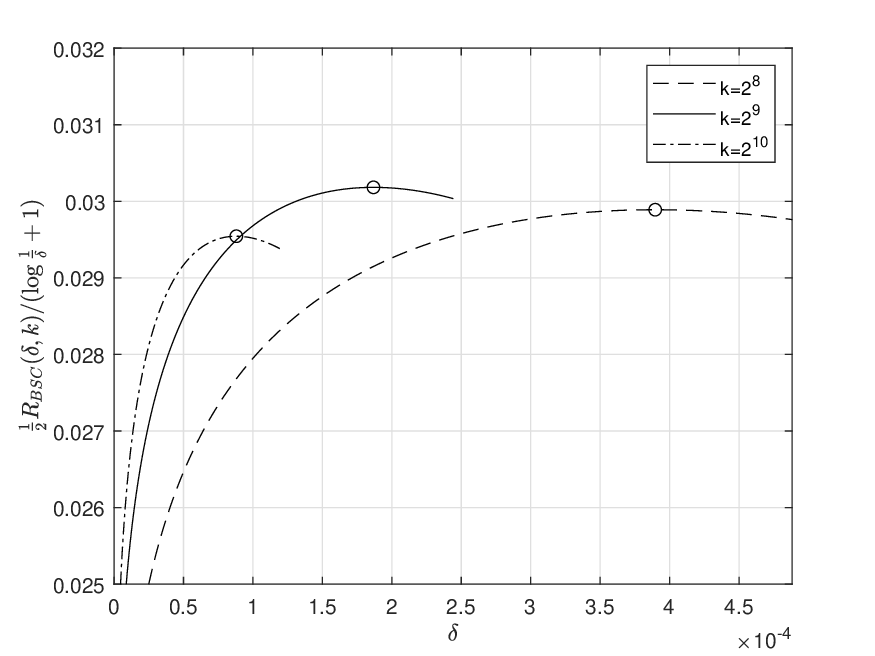} 
	\end{center}
	\caption{ 
	\newtext{Calculating the bound of Theorem~\ref{theorem:interactivebound} using \eqref{eq:ratokdelta} and various combinations of $\levelblock$ and $\targetProb$.
	\label{fig:capratio}
	}}
\end{figure}


\section{Description of the coding scheme for the BSC \label{section:SchemeDescription}}

The \textit{rewind-if-error} scheme is based on two concepts:  uncoded transmission and retransmissions based on error detection. The uncoded transmission is motivated by the fact that 
in a general interactive {protocol}, even in a noise-free environment, the parties cannot predict the {transcript bits} to be output by their counterpart, and hence might not always know some of their own future outputs. For this reason, long blocks of bits, which are essential for efficient block codes, cannot be generated.

The concept of retransmissions based on error detection can be viewed as an extension of the classic example of the {one-way} BEC with feedback \cite[p. 506]{GallagerIT}. In this simple setup, 
channel errors occur independently with probability $\epsilon$ and errors are detected and marked as erasures, whose locations are immediately revealed to both parties. The coding scheme is simply resending the erased bits, yielding an average rate of $1-\epsilon$, which is exactly Shannon's capacity for the BEC. In addition, since all the channel errors are marked as erasures, the probability of decoding error is zero.

When performing interactive communication over a BSC, channel errors are not necessarily marked as erasures and perfect feedback is not present. However, the fact that the parties have (a noisy) two-way communication link, enables them to construct a coding scheme in a similar spirit as follows. The parties start by simulating the {transcript} in a window (or a \textit{block}) of $\levelblock$ consecutive {bits}, operating as if the channel is error free. The probability of error in the window can be upper bounded using the union bound by $\levelblock\BSCprob$, and this number is assumed to be small. Next, the parties exchange bits in order to decide if the window is correct, i.e., no errors occurred, which would lead to the {simulation} of the consecutive window, or incorrect, i.e., some errors occurred, which would lead to retransmission (i.e. {re-simulation of the window}). 

Unfortunately, error detection using less than $\levelblock$ bits of communication has an inherent failure probability. In addition, performing the error detection over a noisy channel can cause further errors, including a disagreement between the parties regarding the mere presence of the errors. For this purpose, the error detection is done in a \textit{hierarchical} and \textit{layered} fashion. Namely, after $\levelblock$ windows are simulated, error detection is applied on all of them, including on the outcome of the previous error detection, possibly initiating their entire retransmission. After $\levelblock^2$ windows are simulated, error detection is applied on all of them, and so on. An illustrated example for this concept for $\levelblock=4$ is given in Table~\ref{table:simulationexample}. 

We are now ready to describe the coding scheme. We note that it can be viewed both as a sequential algorithm and as a recursive algorithm. For sake of clarity and simplicity of exposition, we chose the sequential interpretation for the description and the recursive interpretation for the analysis. 
\subsection{Building blocks}
In the sequel we assume that the order of speakers is bit.~vs.~bit, namely, Alice speaking at odd times and Bob speaking at even times. We denote the input of a the channel by $\channelTrans_\tind$  and its corresponding output by $\channelReceive_\tind$. 
We denote by $\tind$ the time index used for the protocol simulation. $\tind\in\{0,\ntrans-1\}$,
where $\ntrans$ denotes the number of times the channels are used for the simulation of the protocol, excluding the overhead required for the calculation of the rewind bits. In other words, for the sake of simplicity, the instances in which the channels are used for error detection are counted and indexed separately.

The following notions are used as the building blocks of the scheme:
\begin{itemize}
	\item The \textit{uncoded simulation} {of the transcript} is a sequence of bits, generated by the parties and the channel, using the transmission functions in $\bs{\transFunc}$ and disregarding the channel errors. Alice's and Bob's uncoded simulation vectors are for odd $\tind$:
	$\undcodedSim^A_{\tind}\dfn(\channelTrans_{1}, \channelReceive_{2},\ldots ,\channelTrans_{\tind})$, and 
	$\undcodedSim^B_{\tind}\dfn(\channelReceive_{1}, \channelTrans_{2},\ldots ,\channelReceive_{\tind})$ respectively.
	For even $\tind$ they are
	$\undcodedSim^A_{\tind}\dfn(\channelTrans_{1}, \channelReceive_{2},\ldots ,\channelReceive_{\tind})$, and 
	$\undcodedSim^B_{\tind}\dfn(\channelReceive_{1}, \channelTrans_{2},\ldots ,\channelTrans_{\tind})$ respectively.

	\item The \textit{cursor} variables indicate the time indexes of the transmission functions {(i.e. the appropriate function in $\bs{\transFunc}$)} used by Alice or Bob in the previous transmission. We denote Alice's and Bob's cursors by $\cursor^A$ and $\cursor^B$ respectively. We note that $\cursor^A$ and $\cursor^B$ are random variables and may not be identical.
	\item The \textit{rewind bits} are the result of the error detection procedure and are calculated at predetermined points throughout the scheme. 
	They determine whether the {simulation of the transcript} should proceed forward, or rewind.	
	We recall that $\ntrans$ denotes the number of times the channels are used for the simulation of the protocol, excluding the overhead required for the calculation of the rewind bits. We define the number of \textit{layers} by $\nlevels=\log_{\levelblock}\ntrans$, so that
$\ntrans=\levelblock^\nlevels$. We then separate the rewind bits into layers : $\levelind=1,\ldots,\nlevels$. At layer $\levelind$ there are $\levelblock^{\nlevels-\levelind}$ rewind bits, denoted by $\rewindbit^A_{\levelind}(1),...,\rewindbit^A_{\levelind}(\levelblock^{\nlevels-\levelind})$ for Alice and 
$\rewindbit^B_{\levelind}(1),...,\rewindbit^B_{\levelind}(\levelblock^{\nlevels-\levelind})$ for Bob. The value of Alice's and Bob's rewind bits might differ in the general case. 
The rewind bits $\rewindbit^A_{\levelind}(m)$ and $\rewindbit^B_{\levelind}(m)$ are calculated after exactly 
$m\levelblock^{\levelind}$ bits of uncoded simulation, and are calculated according to their respective \textit{rewind windows}.
In the sequel we use the term \textit{active} to denote that a rewind bit is set to one, and \textit{inactive} if it is set to zero.
\item The \textit{rewind window} 
$\rewindwindow{\rewindbit^A_{\levelind}(m)}$ of Alice (resp. $\rewindwindow{\rewindbit^B_{\levelind}(m)}$ of Bob) contains the bits according to which $\rewindbit^A_{\levelind}(m)$ (resp. $\rewindbit^A_{\levelind}(m)$) is calculated. It contains the uncoded simulation bits of the respective party, between times $(m-1)\levelblock^\levelind+1$ and $m\levelblock^\levelind$. In addition it contains all the rewind bits of levels $1\leq \tilde{l}<\levelind$ the party has calculated between these times. 
\end{itemize}

We note, that at every point of the simulation, having the uncoded simulation bits and the rewind bits calculated so far, both parties can calculate their cursors $\cursor^A$ and $\cursor^B$ and their {simulations} of the {transcript}. We denote these {simulation vectors} by: $\hat{\bs{\protTrans}}_A$ and $\hat{\bs{\protTrans}}_B$ for Alice and Bob respectively. We are now ready to introduce the coding scheme.
\subsection*{The coding scheme}
The coding scheme is elaborated in Algorithm~\ref{alg:rewindiferror}. Note that this is the
scheme as implemented at Alice's. The coding scheme implemented at Bob's side is obtained by respectively replacing $\cursor^A$, $\undcodedSim^A$, $\hat{\bs{\protTrans}}_A$, $\rewindbit^A_{\levelind}(m)$, "if $\cursor^A$ is odd", 
$\channelTrans_{\tind} = {\phi}^\text{Alice}_{{\cursor^A}}(\hat{\bs{\protTrans}}_A^{\cursor^A-1})$
by 
$\cursor^B$, $\undcodedSim^B$, $\hat{\bs{\protTrans}}_B$, $\rewindbit^B_{\levelind}(m)$, "if $\cursor^B$ is even",
$\channelTrans_{\tind} = {\phi}^\text{Bob}_{{\cursor^B}}(\hat{\bs{\protTrans}}_B^{\cursor^B-1})$.

\begin{algorithm}
\DontPrintSemicolon
\SetAlgoVlined
\Init{}{
$\tind=0$ \SetFillComment{ {the channel-use index} }\\
$\cursor^A=0$ \SetSideCommentLeft{ {the cursor variable} }\\
$\undcodedSim^A_0=\emptyset$ \SetSideCommentLeft{ {the uncoded simulation vector} }\\
$\hat{\bs{\protTrans}}_A=\emptyset$ \SetSideCommentLeft{ {the transcript simulation vector} }\\
}
\While{$\tind\leq \ntrans$}{
 \SetFillComment{{uncoded simulation of $\levelblock$ bits:} }\\
\For{$\ell=0$ \KwTo $\levelblock$}
{
$\tind = \tind+1$\\
$\cursor^A=\cursor^A+1$\\
\eIf{$\cursor^A$ is odd}
{
$\channelTrans_{\tind} = {\phi}^\text{Alice}_{{\cursor^A}}(\hat{\bs{\protTrans}}_A^{\cursor^A-1})$
\SetFillComment{{produce a transcript bit} }\\
$\undcodedSim^A_\tind=(\undcodedSim^A_{\tind-1}, \channelTrans_{\tind})$
\\
$\hat{\bs{\protTrans}}_A^{\cursor^A}=(\hat{\bs{\protTrans}}_A^{\cursor^A-1},\channelTrans_{\tind})$ 
}
{
receive $\channelReceive_{\tind}$,
\SetFillComment{{the transcript bit produced by Bob} }\\
$\undcodedSim^A_\tind=(\undcodedSim^A_{\tind-1}, \channelReceive_{\tind})$
\\
$\hat{\bs{\protTrans}}_A^{\cursor^A}=(\hat{\bs{\protTrans}}_A^{\cursor^A-1},\channelReceive_{\tind})$ 
}
}
\SetFillComment{{check if a rewind window is full and operate accordingly:} }\\
\For{$\levelind=1$ \KwTo $\nlevels$}
{
\If{$\tind$ mod $\levelblock^\levelind=0$}
{
$m=\tind/\levelblock^\levelind$\\
the rewind window $\rewindwindow{\rewindbit^A_{\levelind}(m)}$ is full\\
calculate the rewind bit $\rewindbit^A_{\levelind}(m)$
by Algorithm~\ref{alg:HammingErrorDetect} if $\levelind=1$
or Algorithm~\ref{alg:PolyErrorDetect} if $\levelind>1$
\\
\eIf{$\rewindbit^A_{\levelind}(m)=1$ rewind}
{
rewind $\cursor^A$ to the value it had at the beginning of $\rewindwindow{\rewindbit^A_{\levelind}(m)}$\\
delete the values of $\rewindwindow{\rewindbit^A_{\levelind}(m)}$ from $\hat{\bs{\protTrans}}_A$\\
set the values corresponding to $\rewindwindow{\rewindbit^A_{\levelind}(m)}$ in 
$\undcodedSim^A_{\tind}$ to zero
}
{
do nothing
}
}
}
}
\caption{\newtext{The coding scheme as implemented by Alice. }
\label{alg:rewindiferror}
}
\end{algorithm}

\subsection*{Calculation of the rewind bits}
For the first layer, $\levelind=1$, the rewind bits are calculated 
using the algorithm for error detection using an extended-Hamming code, described in Definition~\ref{def:exhammingdetection}. The reason for the choice of this procedure is the fact that in the first layer the difference between $\undcodedSim^A$ and $\undcodedSim^B$ is only the channel noise, which is i.i.d. Bernoulli($\BSCprob$), and the fact that the extended-Hamming code 
is a good error detection code for such a noise. In particular, this code is 
\textit{proper} \cite{klove2012error}, which means that the probability of error mis-detection is monotonically increasing for $0<\BSCprob<1/2$. As the probability of mis-detection for $\BSCprob=\frac{1}{2}$ is equal to that of random hashing with the same number of bits, for $\BSCprob<\frac{1}{2}$ we obtain favorable performance without randomness.
The details of the calculation are elaborated in Algorithm~\ref{alg:HammingErrorDetect}.
\begin{algorithm}[h]
\DontPrintSemicolon
\KwIn{\\$\undcodedSim^A=\rewindwindow{\rewindbit^A_{1}(m)}$ - Alice's {rewind window}, $\levelblock$ bits row-vector\\
$\undcodedSim^B=\rewindwindow{\rewindbit^B_{1}(m)}$ - Bob's {rewind window}, $\levelblock$ bits row-vector \\
$H$ - the parity check matrix of a $(\levelblock,\levelblock-\log\levelblock-1,4)$ {extended-Hamming code}
}
\KwOut{
\\${\rewindbit^A_{1}(m)}$ - Alice's {rewind bit}\\
${\rewindbit^B_{1}(m)}$ - Bob's {rewind bit}\\
}
\Alg{}
{
\underline{Alice}: calculate $\syndromevec^A =  \undcodedSim^A\cdot H^T$ \\
\underline{Alice}: send $\syndromevec^A$ to Bob over the channel using $\repetitioncoeffH$ repetitions per bit\\
\underline{Bob}: decode $\hat{\syndromevec}^A$ using a majority vote per bit on the channel respective inputs\\
\underline{Bob}: calculate $\syndromevec^B = \undcodedSim^B \cdot H^T$ \\
\underline{Bob}: $\rewindbit^B_{1}(m)=\indfunc{\hat{\syndromevec}^A\neq \syndromevec^B}$\\
\underline{Bob}: send $\rewindbit^B_{1}(m)$ to Alice over the channel using $\repetitioncoeffH$ repetitions per bit\\
\underline{Alice}: Set $\rewindbit^A_{1}(m)$ according to a  majority vote per bit on the channel respective input
}
\SetAlgoVlined
\caption{\newtext{Calculating of the rewind bits at $\levelind=1$\label{alg:HammingErrorDetect}}}
\end{algorithm}

For all other layers, $\levelind>1$, the procedure is implemented according to 
the polynomial based randomized error detection scheme from Definition~\ref{def:equalityPrivate}. We start by assuming that the parties agree on the prime number $q_\levelind$ for every layer $\levelind>1$. We also assume for simplicity of exposition, that for every rewind window, the parties commonly and independently draw a test point $U$ using a common random string.
We denote the set comprising all the test points used by the scheme by $\mathcal{U}$, which contains $|\mathcal{U}|=O(\protLength)$ elements.
 In Section~\ref{section:derandomize} we show how the common randomness assumption can be relaxed. 
 The details of the calculation are elaborated in Algorithm~\ref{alg:PolyErrorDetect}.

\begin{algorithm}[h]
\DontPrintSemicolon
\KwIn{\\$\undcodedSim^A=\rewindwindow{\rewindbit^A_{1}(m)}$ - Alice's {rewind window}, $\levelblock^{\levelind}$ bits row-vector\\
$\undcodedSim^B=\rewindwindow{\rewindbit^B_{1}(m)}$ - Bob's {rewind window}, $\levelblock^{\levelind}$ bits row-vector \\
$H$ - the parity check matrix of a $(\levelblock,\levelblock-\log\levelblock-1,4)$ {extended-Hamming code}
}
\KwOut{
\\${\rewindbit^A_{1}(m)}$ - Alice's {rewind bit}\\
${\rewindbit^B_{1}(m)}$ - Bob's {rewind bit}\\
}
\Alg{}
{
\underline{Alice \& Bob}: uniformly draws $U\in \GFq$, $U\neq 0$ .\\
\underline{Alice}: 
calculates $A(U,\undcodedSim^A)=\sum_{i=1}^\ell X_i^A U^{i-1} (\text{mod } q)$\\
\underline{Alice}: 
send the bits representing $A(U,\undcodedSim^A)$ over the channel to Bob using 
$\repetitioncoeff+2\levelind$ repetitions per bit\\
\underline{Bob}: decode $\tilde{A}(U,\rewindwindow{\rewindbit^A_{\levelind}(m)})$
using a majority vote per bit on the channel respective inputs\\
\underline{Bob}: calculate $B(U,\undcodedSim^A)=\sum_{i=1}^\ell X_i^B U^{i-1} (\text{mod } q)$ \\
\underline{Bob}: $\rewindbit^{B}_{\levelind}(m)=\indfunc{\tilde{A}({U},\undcodedSim^A)\neq B({U},\undcodedSim^B)}$\\
\underline{Bob}: send $\rewindbit^B_{\levelind}(m)$ to Alice over the channel using $\repetitioncoeff+2\levelind$ repetitions per bit\\
\underline{Alice}: set $\rewindbit^A_{\levelind}(m)$ according to a  majority vote per bit on the channel respective input
}
\SetAlgoVlined
\caption{
\newtext{Calculating of the rewind bits at $\levelind>1$\label{alg:PolyErrorDetect}}}
\end{algorithm}


Let us now bound the number of bits required for this procedure. 
First, we generously bound the number of bits in a rewind window of layer $\levelind$, which contains all the uncoded simulation bits and the nested rewind bits of the previous layers, by $2\levelblock^{\levelind}$. For layer $\levelind$, the parties set $q_\levelind$ to be the first prime number between $2\levelblock^{2+\levelind}$ and $4\levelblock^{2+\levelind}$.
Therefore, a number in $\GFql$ can be represented by no more than $2+(2+\levelind)\log\levelblock$ bits. All in all the procedure described above required $3+(2+\levelind)\log\levelblock$ bits for layer $\levelind$. For simplicity of calculation, from this point on, we bound this number by
\begin{align}\label{eq:errordetectinobits}
3+(2+\levelind)\log \levelblock< 3\levelind\log \levelblock,
\end{align}
which applies for any $\levelind\geq 2$ and $\levelblock\geq 4$.


\begin{table*}
\underline{Start the simulation:} Initialize the cursors: $\cursor^A=\cursor^B=0$\\ \\
		\begin{tabular}{c|c|c|}\cline{2-3}
			&{ $\rewindwindow{\rewindbit_{1}(1)}$}&$\rewindbit_{1}(1)$\\ \hline
			\multicolumn{1}{|c|}{A}&$0,0,1,1$ &\multicolumn{1}{|c|}{$0$}\\ \hline
			\multicolumn{1}{|c|}{B}&$0,0,1,1$&\multicolumn{1}{|c|}{$0$}\\ \hline
		\end{tabular}
\\ \\
\underline{End of $\rewindwindow{\rewindbit_{1}(1)}$:} No errors, continue. $\cursor^A=\cursor^B=4$
\\ \\
		\begin{tabular}{c|c|c|c|c|}\cline{2-5}
			&{$\rewindwindow{\rewindbit_{1}(1)}$}&$\rewindbit_{1}(1)$ & {$\rewindwindow{\rewindbit_{1}(2)}$}&$\rewindbit_{1}(2)$\\ \hline
			\multicolumn{1}{|c|}{A}&$0,0,1,1$ &\multicolumn{1}{|c|}{$0$} & $1,0,0,1$ &$0$ \\ \hline
			\multicolumn{1}{|c|}{B}&$0,0,1,1$&\multicolumn{1}{|c|}{$0$}&$1,0,0,1$&$0$\\ \hline
		\end{tabular}
\\ \\
\underline{End of $\rewindwindow{\rewindbit_{1}(2)}$:} No errors, continue. $\cursor^A=\cursor^B=8$
\\	\\
		\begin{tabular}{c|c|c|c|c|c|c|}\cline{2-7}
			&{$\rewindwindow{\rewindbit_{1}(1)}$}&$\rewindbit_{1}(1)$ & {$\rewindwindow{\rewindbit_{1}(2)}$}&$\rewindbit_{1}(2)$
			& {$\rewindwindow{\rewindbit_{1}(3)}$}&$\rewindbit_{1}(3)$\\ \hline
			\multicolumn{1}{|c|}{A}&$0,0,1,1$ &\multicolumn{1}{|c|}{$0$} & $1,0,0,1$ &$0$&
			$0,\detectederror{0},0,0$&${1}$ \\ \hline
			\multicolumn{1}{|c|}{B}&$0,0,1,1$&\multicolumn{1}{|c|}{$0$}&$1,0,0,1$&${0}$&
			$0,\detectederror{1},0,0$&${1}$\\ \hline
		\end{tabular}
\\ \\
\underline{End of $\rewindwindow{\rewindbit_{1}(3)}$:} An error occurred and was detected by both parties: $\rewindbit^A_{1}(3)=\rewindbit^B_{1}(3)=1$\\
 Both parties zero the rewind window and rewind the cursors to the value it had before the window started: $\cursor^A=\cursor^B=8$
\\	\\
		\begin{tabular}{c|c|c|c|c|c|c|c|c|}\cline{2-9}
			&{$\rewindwindow{\rewindbit_{1}(1)}$}&$\rewindbit_{1}(1)$ & {$\rewindwindow{\rewindbit_{1}(2)}$}&$\rewindbit_{1}(2)$ & {$\rewindwindow{\rewindbit_{1}(3)}$}&$\rewindbit_{1}(3)$ &
			{$\rewindwindow{\rewindbit_{1}(4)}$}&$\rewindbit_{1}(4)$
			\\ \hline
			\multicolumn{1}{|c|}{A}&$0,0,1,1$ &\multicolumn{1}{|c|}{$0$} & $1,0,0,1$ &$0$&
			$\zeroedbits{0},\zeroedbits{0},\zeroedbits{0},\zeroedbits{0}$&${1}$ &$0,1,1,1$&$\detectederror{1}$\\ \hline
			\multicolumn{1}{|c|}{B}&$0,0,1,1$&\multicolumn{1}{|c|}{$0$}&$1,0,0,1$&${0}$&
			$\zeroedbits{0},\zeroedbits{0},\zeroedbits{0},\zeroedbits{0}$&${1}$			&$0,1,1,1$&$\detectederror{0}$\\ \hline
		\end{tabular}
\\ \\
\underline{End of $\rewindwindow{\rewindbit_{1}(4)}$:} There are no errors so Bob calculates $\rewindbit^B_{1}(3)=0$ and continues ($\cursor^B=12$). \\
However due to an error in communicating $\rewindbit^B_{1}(3)$, Alice decodes $\rewindbit^A_{1}(3)=1$, zeros the window and rewinds the cursor ($\cursor^A=8$)
\\	\\
		\begin{tabular}{c|c|c|c|c|c|c|c|c|c|}\cline{2-10}
			&\multicolumn{8}{|c|}{$\rewindwindow{\rewindbit_{2}(1)}$}&\\ \cline{2-10}
			&{$\rewindwindow{\rewindbit_{1}(1)}$}&$\rewindbit_{1}(1)$ & {$\rewindwindow{\rewindbit_{1}(2)}$}&$\rewindbit_{1}(2)$ & {$\rewindwindow{\rewindbit_{1}(3)}$}&$\rewindbit_{1}(3)$ &
			{$\rewindwindow{\rewindbit_{1}(4)}$}&$\rewindbit_{1}(4)$&$\rewindbit_{2}(1)$
			\\ \hline
			\multicolumn{1}{|c|}{A}&$0,0,1,1$ &\multicolumn{1}{|c|}{$0$} & $0,0,0,0$ &$0$&
			$0,0,0,0$&$0$ &$\zeroedbits{0},\detectederror{\zeroedbits{0}},\detectederror{\zeroedbits{0}},\detectederror{\zeroedbits{0}}$&$\detectederror{1}$&$1$\\ \hline
			\multicolumn{1}{|c|}{B}&$0,0,1,1$&\multicolumn{1}{|c|}{$0$}&$1,0,0,1$&${0}$&
			$0,0,0,0$&$0$&$0,\detectederror{1},\detectederror{1},\detectederror{1}$&$\detectederror{0}$&$1$\\ \hline
		\end{tabular}
\\ \\ 
\underline{ End of $\rewindwindow{\rewindbit_{2}(1)}$}. Calculate $\rewindbit_{2}(1)$.\\
The errors are detected so $\rewindbit^A_{2}(1)=\rewindbit^B_{2}(1)=1$, and the cursors are rewound to the beginning of the window : $\cursor^A=\cursor^B=0$.
\\	\\
\begin{tabular}{c|c|c|c|c|c|c|c|c|c|c|c|}\cline{2-10}
	&\multicolumn{8}{|c|}{$\rewindwindow{\rewindbit_{2}(1)}$}&\\ \cline{2-12}
	&{$\rewindwindow{\rewindbit_{1}(1)}$}&$\rewindbit_{1}(1)$ & {$\rewindwindow{\rewindbit_{1}(2)}$}&$\rewindbit_{1}(2)$ & {$\rewindwindow{\rewindbit_{1}(3)}$}&$\rewindbit_{1}(3)$ &
	{$\rewindwindow{\rewindbit_{1}(4)}$}&$\rewindbit_{1}(4)$&$\rewindbit_{2}(1)$&
	$\rewindwindow{\rewindbit_{1}(5)}$&${\rewindbit_{1}(5)}$
	\\ \hline
	\multicolumn{1}{|c|}{A}&$\zeroedbits{0},\zeroedbits{0},\zeroedbits{0},\zeroedbits{0}$&$\zeroedbits{0}$&
	$\zeroedbits{0},\zeroedbits{0},\zeroedbits{0},\zeroedbits{0}$&$\zeroedbits{0}$&
	$\zeroedbits{0},\zeroedbits{0},\zeroedbits{0},\zeroedbits{0}$&$\zeroedbits{0}$&
	$\zeroedbits{0},\zeroedbits{0},\zeroedbits{0},\zeroedbits{0}$&$\zeroedbits{0}$&
	${1}$&$0,0,1,1$&$0$\\ \hline
	\multicolumn{1}{|c|}{B}&$\zeroedbits{0},\zeroedbits{0},\zeroedbits{0},\zeroedbits{0}$&$\zeroedbits{0}$&
$\zeroedbits{0},\zeroedbits{0},\zeroedbits{0},\zeroedbits{0}$&$\zeroedbits{0}$&
$\zeroedbits{0},\zeroedbits{0},\zeroedbits{0},\zeroedbits{0}$&$\zeroedbits{0}$&
$\zeroedbits{0},\zeroedbits{0},\zeroedbits{0},\zeroedbits{0}$&$\zeroedbits{0}$&
${1}$&$0,0,1,1$&$0$\\ \hline
\end{tabular}
\\ \\ 
\underline{End of $\rewindwindow{\rewindbit_{1}(5)}$} The first four bits of the {protocol} are re-simulated. No errors. $\cursor^A=\cursor^B=4$.
\caption{Example for a rewind-if-error coding scheme with $\levelblock=4$. Detected error are in $\detectederror{{bold}}$, zeroed bits are in $\zeroedbits{{blue}}$. \label{table:simulationexample}}
\end{table*}


\section{Analysis of the coding scheme \label{section:schemeanalysis}: A Proof of Theorem~\ref{theorem:rewindifrateHamming}}
We start by giving the following notation:
\begin{itemize}
    \item $\ntrans$ is the number of times the channels are used for the protocol simulation, including retransmissions and excluding the overhead required for the transmission of the rewind bits.
	\item $\cursormin\dfn\min\{\cursor^A,\cursor^B\}$ is the minimum between Alice's and Bob's cursor at any moment
	\item $\cursormin({\ntrans}), \cursor^A({\ntrans}), \cursor^B({\ntrans})$ denote the respective values of 
	$\cursormin,\cursor^A, \cursor^B$ at the end of the simulation
	\item $\hat{\bs{\protTrans}}_A^{\cursormin({\ntrans})}$ and $\hat{\bs{\protTrans}}_B^{\cursormin({\ntrans})}$ denote the first ${\cursormin({\ntrans})}$ bits of Alice's and Bob's {simulations of the transcript} respectively, at the end of the simulation. We also assume that if $\cursor^{A}({\ntrans})>\protLength$ or $\cursor^{B}({\ntrans})>\protLength$ then the parties proceed the {protocol} by transmitting zeros 
	\item We denote $\rewindbitOr_{\levelind}(m)\dfn\rewindbit^A_{\levelind}(m)\vee\rewindbit^B_{\levelind}(m)$. Namely, $\rewindbitOr_{\levelind}(m)$ it is defined as the disjunction between Alice's and Bob's respective rewind bits
\end{itemize}
The following two error events will be analyzed
\begin{itemize}
\item $\mathcal{E}_1$ is the event in which $\cursormin({\ntrans})< \protLength$
\item $\mathcal{E}_2$ is the event in which either $\hat{\bs{\protTrans}}_A^{\cursormin({\ntrans})}\neq\bs{\protTrans}^{\cursormin({\ntrans})}$ or $\hat{\bs{\protTrans}}_B^{\cursormin({\ntrans})}\neq\bs{\protTrans}^{\cursormin({\ntrans})}$ 
\end{itemize}
The simulation error event is included in $\mathcal{E}_1\cup\mathcal{E}_2$ and we would like its respective probability to vanish with $\protLength$.

We start by analyzing $\Pr(\mathcal{E}_1)$ and do it by lower bounding $\cursormin({\ntrans})$. We recall that by construction of the scheme, $\rewindbit^A_{\levelind}(m)=1$ (resp. $\rewindbit^B_{\levelind}(m)=1$) will rewind $\cursor^A$ (resp. $\cursor^B$) to the value it had at the beginning of the rewind window. Namely $\cursor^A$ (resp. $\cursor^B$) will be reduced by at most $\levelblock^{\levelind}$. It is now instrumental to use the definitions of $\cursormin$ and $\rewindbitOr_{\levelind}(m)$ and observe that if either $\rewindbit^A_{\levelind}(m)=1$ or $\rewindbit^B_{\levelind}(m)=1$ (namely, if $\rewindbitOr_{\levelind}(m)=1$) then the minimal among $\cursor^A$ and $\cursor^B$ (namely, $\cursormin$) will be reduced by at most $\levelblock^{\levelind}$. 
Recalling that $\ntrans=\levelblock^{\nlevels}$ we can now write
\begin{align}
\cursormin(\ntrans)&\geq\ntrans-\sum_{\levelind=1}^\nlevels\sum_{m=1}^{\levelblock^{\nlevels-\levelind}}\rewindbitOr_{\levelind}(m)\levelblock^{\levelind}\\
&=\ntrans\left(1-\sum_{\levelind=1}^{\nlevels}\averewind\right),\label{eq:cursorA}
\end{align}
where 
\begin{align}\label{eq:blbardef}
\averewind\dfn
\frac{\sum_{m=1}^{\levelblock^{\nlevels-\levelind}}\rewindbitOr_{\levelind}(m)}
{\levelblock^{\nlevels-\levelind}}
\end{align}	
denotes the average number of active (i.e., nonzero) rewind bits at level $\levelind$.
We note that by construction of the scheme (including its use of randomness), the processes of the error generation and detection are identical for all blocks at level $\levelind$. For this reason, the probability of having an active rewind bit is also identical for all the blocks at level $\levelind$. We denote this probability by
\begin{align}
\Pbl = \Pr(\rewindbitOr_{\levelind}(1)=1)=...=
\Pr(\rewindbitOr_{\levelind}(\levelblock^{\nlevels-\levelind})=1).
\end{align}
Taking the expectation over \eqref{eq:cursorA} yields
\begin{align}\label{eq:expcursor}
\Expt\cursormin(\ntrans)\geq \ntrans\left(1-\sum_{\levelind=1}^\nlevels \Pbl\right).
\end{align}
In order to proceed with the calculation of $\Pbl$, we define $\Pfailure_{\levelind}$ as the probability that either $\rewindbit^A_{\levelind}(m)$ or $\rewindbit^B_{\levelind}(m)$ differ from the error indicator 
$\indfunc{\rewindwindow{\rewindbit^A_{\levelind}(m)}\neq\rewindwindow{\rewindbit^A_{\levelind}(m)}}$. This probability does not depend on $m$ due to the same considerations as above.

The following lemma bounds $\Pfailure_{\levelind}$:
\begin{lemma}\label{lemma:Pfailurebound}
For $\levelind=1$
\begin{align}
\Pfailure_{1}&\leq 
\frac{1}{2\levelblock}\left(1+2(\levelblock-1)(1-2\BSCprob)^{\frac{\levelblock}{2}}+(1-2\BSCprob)^\levelblock\right)\\
&-(1-\BSCprob)^\levelblock
+(3+\log\levelblock)\batdistance^{\repetitioncoeffH},
\end{align} 
and for $\levelind>1$
\begin{align}\label{eq:Pfailurel}
\Pfailure_{\levelind}&\leq\\
&\levelblock^{-\levelind}
\left(
\levelblock\Pfailure_{1}
+3\batdistance^{\repetitioncoeff+4}\levelblock^{2}\log\levelblock
\frac{2-\batdistance^2\levelblock}{(1-\batdistance^2\levelblock)^2}\right).
\end{align}
\end{lemma}

\begin{proof}
For the first layer
\begin{align}\label{eq:Pfailure1H}
\Pfailure_{1}&\leq\\
&\Pr\left(NEQ=0, \undcodedSim^A\neq\undcodedSim^B\right)+(3+\log\levelblock)\batdistance^{\repetitioncoeffH},
\end{align}
where $\Pr\left(NEQ=0, \undcodedSim^A\neq\undcodedSim^B\right)$ is the error mis-detection probability of the extended-Hamming code based error detection scheme of Definition~~\ref{def:exhammingdetection} as given in \eqref{eq:hammingmisdetect}. $\batdistance^{\repetitioncoeffH}$ is the probability of error in the decoding of a bit sent with $\repetitioncoeffH$ repetitions according to Lemma~\ref{lemma:repetition}. The multiplication by $(3+\log\levelblock)$ accounts for the union bound over the number of bits used for the error detection: $2+\log\levelblock$ bits sent from Alice to Bob ($1+\log\levelblock$ required for the description of the syndrome according to Lemma~\ref{def:exhammingdetection} and an additional bit reserved for avoiding the random tie breaking as described in Subsection~\ref{subsection:tieerasure}) and a single bit fed back from Bob to Alice. 

The key idea in the analysis of the scheme for $\levelind>1$ is regarding the calculation of the rewind bits as a \textit{layered} recursive process. Namely, we observe that by construction, a rewind window at level $\levelind$ comprises $\levelblock$ rewind windows of level $\levelind-1$. In addition, the polynomial based randomized error detection of Definition~\ref{def:equalityPrivate} uses independent test points for every layer and hence is independent between layers.
Having this notion we can write the following recursion formula:
\begin{align}\label{eq:PelRecurssionBasic}
\hspace{-5mm}
\Pfailure_{\levelind}&\leq \levelblock^{-2}\levelblock\Pfailure_{\levelind-1}+(2+(2+\levelind)\log\levelblock)\batdistance^{\repetitioncoeff+2\levelind}
\end{align}
where $\levelblock\Pfailure_{\levelind-1}$ is the union bound over th error events of the previous level.
The multiplication by $\levelblock^{-2}$ accounts for the probability of error mis-detection according to Lemma~\ref{def:equalityPrivate} with the setting $\Pmd=\levelblock^{-2}$ and $\ell$ as the number of bits in the appropriate rewind window $\rewindwindow{\rewindbit^A_{\levelind}(m)}$ (or $\rewindwindow{\rewindbit^B_{\levelind}(m)}$). As described above, for the error detection, Alice should send Bob a number in $\GFq$ and Bob should reply with a single bit (we assume that the set of test points $\mathcal{U}$ is jointly drawn by the parties using common randomness). We recall that the number of bits required for the error detection scheme of Definition~\ref{def:equalityPrivate} is generously bounded by $3\levelind\log \levelblock$ due to \eqref{eq:errordetectinobits}.
 All in all, we can rewrite \eqref{eq:PelRecurssionBasic} as
\begin{align}
\Pfailure_{\levelind}\leq\levelblock^{-1}\Pfailure_{\levelind-1}+
3\batdistance^{\repetitioncoeff}(\log\levelblock)\levelind\batdistance^{2\levelind}
\label{eq:PelRecurssion}
\end{align}
Solving the recursion of \eqref{eq:PelRecurssion} with the initial condition in \eqref{eq:Pfailure1H} we can  bound $\Pfailure_{\levelind}$ as follows:
\begin{align}
&\Pfailure_{\levelind}
\leq \levelblock^{1-\levelind}\Pfailure_{1}
+3\batdistance^{\repetitioncoeff}(\log\levelblock)\sum_{j=2}^\levelind j \batdistance^{2j}
\levelblock^{j-\levelind}\\
&\leq \levelblock^{1-\levelind}\Pfailure_{1}
+3\batdistance^{\repetitioncoeff}(\log\levelblock)\levelblock^{-\levelind}
\sum_{j=2}^{\infty} j (\batdistance^{2}
\levelblock)^{j}\label{eq:sumbeta}\\
&= 
\levelblock^{1-\levelind}\Pfailure_{1}
+3\batdistance^{\repetitioncoeff}(\log\levelblock)
\levelblock^{-\levelind}
(\batdistance^{2}\levelblock)^2
\frac{2-\batdistance^2\levelblock}{(1-\batdistance^2\levelblock)^2}
\\
&= \levelblock^{-\levelind}
\left(
\levelblock\Pfailure_{1}
+3\batdistance^{\repetitioncoeff+4}\levelblock^{2}\log\levelblock
\frac{2-\batdistance^2\levelblock}{(1-\batdistance^2\levelblock)^2}
\right)\label{eq:PelH}.
\end{align}
We note that the assumption in Theorem~\ref{theorem:rewindifrateHamming} that $\BSCprob<1/(8\levelblock)$ ensures that 
$\batdistance^2\levelblock<1$ ensuring that the infinite sum in \eqref{eq:sumbeta} converges.
\end{proof}		
We are now ready to bound $\Pbl$. We recall that it is defined as the probability that 
either $\rewindbit^A_{\levelind}(m)=1$ or $\rewindbit^B_{\levelind}(m)=1$, and is independent of $m$ due to the symmetry of the scheme. For $\levelind=1$ we use the union bound over the probability of an erroneous bit and a communication error:
\begin{align}
\Pbone\leq  \levelblock\BSCprob+(3+\log\levelblock)\batdistance^{\repetitioncoeffH}
\dfn \Pbonebound\label{eq:PboueboundH}.
\end{align}
Similarly, for $\levelind>1$ we take the union bound over the probability of error $\Pfailure_{\levelind-1}$, in one of the $\levelblock$ blocks in the layer $\levelind-1$ and a communication error: 
\begin{align}
\Pbl&\leq  \levelblock\Pfailure_{\levelind-1}+
3\batdistance^{\repetitioncoeff}(\log\levelblock)\levelind\batdistance^{2\levelind}\\
&\leq 
\levelblock^{2-\levelind}
\left(
\levelblock\Pfailure_{1}
+3\batdistance^{\repetitioncoeff+4}\levelblock^{2}\log\levelblock
\frac{2-\batdistance^2\levelblock}{(1-\batdistance^2\levelblock)^2}
\right)\\
&+3\batdistance^{\repetitioncoeff}(\log\levelblock)\levelind\batdistance^{2\levelind}\\
&\dfn\Pblbound\label{eq:PblboundH}.
\end{align}
Let us now bound the average rewind by
\begin{align}\label{eq:javerage}
\Expt&\cursormin(\ntrans)\geq \ntrans\left(1-\sum_{\levelind=1}^\infty \Pblbound\right)= \ntrans\zeta.
\end{align}
where
\begin{align}
\zeta &\dfn1-\sum_{\levelind=1}^\infty \Pblbound\\
 & = 1- \levelblock\BSCprob-(3+\log\levelblock)\batdistance^{\repetitioncoeffH}\\
&-
\frac{\levelblock^2}{\levelblock-1}
\left(
\Pfailure_{1}
+3\batdistance^{\repetitioncoeff+4}\levelblock\log\levelblock
\frac{2-\batdistance^2\levelblock}{(1-\batdistance^2\levelblock)^2}
\right)\\
&-
3\batdistance^{\repetitioncoeff+4}\log\levelblock
\frac{2-\batdistance^2}{(1-\batdistance^2)^2}.
\end{align}
Setting
\begin{align}\label{eq:ntransprotlength}
\ntrans=\frac{\protLength}{1-\sum_{\levelind=1}^\infty \Pblbound-\xi}=
\frac{\protLength}{\zeta-\xi}
\end{align}
for some $0<\xi<\zeta$ will therefore ensure that $\Expt\cursormin(\ntrans)\geq\protLength$. The following lemma ensures {that} $\Pr(\mathcal{E}_1)$ also vanishes in $\protLength$:
\begin{lemma}\label{lemma:concnetration}
	For any $\xi>0$ and $\ntrans$ that satisfies \eqref{eq:ntransprotlength}:
	\begin{align}\label{eq:concentrationR}
	\lim_{\protLength\to \infty}\Pr(\mathcal{E}_1)=\lim_{\protLength\to \infty}\Pr\left(\cursormin(\ntrans)< \protLength\right)=0.
	\end{align}
\end{lemma}
The proof is in Appendix~\ref{appendix:concentration}. It is based on the fact that
due to \eqref{eq:javerage} and \eqref{eq:ntransprotlength} we have $\Expt\cursormin(\ntrans)\geq(1+\eta) \protLength$ for some $\eta>0$ and using standard concentration techniques. We note that the proof assumes the number of test points in $\mathcal{U}$ is $|\mathcal{U}|=O(\sqrt{\protLength})$, whereas so far we assumed that every use of the error detection procedure of Definition~\ref{def:equalityPrivate} uses a different test point ({i.e. $|\mathcal{U}|=O(\protLength)$}). Since $|\mathcal{U}|=O(\sqrt{\protLength})$ is restrictive,
Lemma~\ref{lemma:concnetration} also holds for the current description of the scheme. The motivation for reducing $|\mathcal{U}|$ is changing the common randomness to private randomness, which is extracted from the channel, and is elaborated in Section~\ref{section:derandomize}.

The following lemma ensures $\Pr(\mathcal{E}_2)$ vanishes in $\protLength$:
\begin{lemma}\label{lemma:PeFinal}
For any $\xi>0$ and $\ntrans$ that satisfies \eqref{eq:ntransprotlength} 
\begin{align}
\lim_{\protLength\to \infty}\Pr(\mathcal{E}_2)=0.
\end{align}
\end{lemma}
\begin{proof}
We remind the reader that $\Pfailure_{\levelind}$ is defined as the probability that either $\rewindbit^A_{\levelind}(m)$ or $\rewindbit^B_{\levelind}(m)$ differ from the error indicator 
$\indfunc{\rewindwindow{\rewindbit^A_{\levelind}(m)}\neq\rewindwindow{\rewindbit^A_{\levelind}(m)}}$. Namely, it is the probability of an undetected error, or a falsely detected error, in the simulation of a block in layer $\levelind$ at least at one party. Since $\nlevels$ is the final layer, and due to the recursive structure of the error detection, $\Pfailure_{\nlevels}$ therefore upper bounds the respective probability at the end of the coding scheme. The error event related to $\Pfailure_{\nlevels}$ includes $\mathcal{E}_2$ and therefore $\Pr(\mathcal{E}_2)\leq\Pfailure_{\nlevels}$. Rewriting \eqref{eq:Pfailurel} and setting $\levelind=\nlevels=\log_\levelblock\ntrans=\log_\levelblock(\protLength/(\zeta-\xi))$ we obtain:
	\begin{align}
	&\Pr(\mathcal{E}_2)\\
	&\leq
	\frac{\zeta-\xi}{\protLength}
	\left(
	\levelblock\Pfailure_{1}
	+3\batdistance^{\repetitioncoeff+4}\levelblock^{2}\log\levelblock
	\frac{2-\batdistance^2\levelblock}{(1-\batdistance^2\levelblock)^2}\right).
	\end{align}
	Therefore $\lim_{\protLength\to \infty}\Pr(\mathcal{E}_2)=0$.
\end{proof}

Let us now bound $\sigmalength$, the total number of channel uses consumed by the scheme:
\begin{align}
&\sigmalength\leq\ntrans+\repetitioncoeffH(3+\log\levelblock)k^{\nlevels-1}\\
&+3\log\levelblock\sum_{\levelind=2}^{\infty}\levelind(\repetitioncoeff+2\levelind)\levelblock^{\nlevels-\levelind}
\\
&\leq\ntrans\left(1+\tfrac{\repetitioncoeffH(3+\log\levelblock)}{\levelblock}\right.\\
&+3\log\levelblock\left.\left[
\tfrac{\repetitioncoeff(2\levelblock-1)}{k(\levelblock-1)^2}
+\tfrac{4\levelblock}{(\levelblock-1)^3}
+\tfrac{4\levelblock-2}{\levelblock(\levelblock-1)^2}
\right]
\right),\label{eq:errordetoverhead}
\end{align}
where $\repetitioncoeffH(3+\log\levelblock)k^{\nlevels-1}$ is the number of channel uses required for the error detection at the first layer, and $3\log\levelblock\sum_{\levelind=2}^{\infty}\levelind(\repetitioncoeff+2\levelind)\levelblock^{\nlevels-\levelind}$ is the number of channel uses required for the error detection in all other layers.
Using \eqref{eq:ntransprotlength} and \eqref{eq:errordetoverhead} we can bound the total rate of the scheme by the term in \eqref{eq:totalratebound}.
\begin{figure*}[!hb]
\rule{\textwidth}{1pt}
\begin{align}\label{eq:totalratebound}
\RBsc(\BSCprob,\levelblock)\geq
\frac
{
1- \levelblock\BSCprob-(3+\log\levelblock)\batdistance^{\repetitioncoeffH}
-\frac{\levelblock^2}{\levelblock-1}
\left(
\Pfailure_{1}
+3\batdistance^{\repetitioncoeff+4}\levelblock\log\levelblock
\frac{2-\batdistance^2\levelblock}{(1-\batdistance^2\levelblock)^2}
\right)
-
3\batdistance^{\repetitioncoeff+4}\levelblock^{2}\log\levelblock
\frac{2-\batdistance^2\levelblock}{(1-\batdistance^2\levelblock)^2}
-\xi
}
{
1+\frac{\repetitioncoeffH(3+\log\levelblock)}{\levelblock}
+3\log\levelblock
\left[
\frac{\repetitioncoeff(2\levelblock-1)}{k(\levelblock-1)^2}
+\frac{4\levelblock}{(\levelblock-1)^3}
+\frac{4\levelblock-2}{\levelblock(\levelblock-1)^2}
\right]
}.
\end{align}
\end{figure*}
Since this term holds for and $\xi>0$, we can take the limit $\xi\to 0$.
Setting $\repetitioncoeff=3$ and $\repetitioncoeffH=5$ (which are the results of an exhaustive search over various possible values) provides 
\eqref{eq:hammingrate} and conclude the proof of  Theorem~\ref{theorem:rewindifrateHamming}. 


\section{Generalization to binary memoryless symmetric channels\label{section:bms}}
In Definition~\ref{def:BMS} we defined a \textit{binary memoryless symmetric} (BMS) channel as a collection of BSC's with various crossover probabilities. 


We now extend the notion of repetition coding of Lemma~\ref{lemma:repetition} to BMS channels.
\begin{definition}\label{def:rhorepetition}[$\rho$-repetition channel]
Let $P^{(\rho)}_{\tilde{Y}|\tilde{X}}$ be the $\rho$-repetition channel corresponding to a BMS($P_{Y|X}$) channel. It is obtained by using 
the bit $\tilde{X}$ as the input of BMS($P_{Y|X}$), $\rho$ consecutive times, hence producing the series of channel outputs $Y_1,...,Y_{\rho}$. The output of $P^{(\rho)}_{\tilde{Y}|\tilde{X}}$ is then calculated using the following equation
\begin{align}\label{eq:MLdecrule}
\tilde{Y}=\argmax_{x\in\{0,1\}}\prod_{i=1}^\rho P_{Y|X=x}(Y_i),
\end{align}
where ties are broken by drawing a $\Ber(1/2)$ random variable  \footnote{If Y is continuous, replace $P_{Y|X=x}$ with the conditional density.}.
\end{definition}
We note that like in the BSC case, we randomly break the ties in order to facilitate the analysis and later explain in Subsection~\ref{subsection:ties} how this random procedure can be circumvented. The following lemma bounds the decoding error of the $\rho$-repetition channel.
\begin{lemma}\label{lemma:BMSrepetition}
	For any BMS($P_{Y|X}$) channel with Shannon capacity $\Cshannon(P_{Y|X})=C$ the corresponding $\rho$-repetition channel $P^{(\rho)}_{\tilde{Y}|\tilde{X}}$ is a $\mathrm{BSC}(\delta)$ with $\delta\leq \batdistance^{\rho}$, where 
$\batdistance=2\sqrt{
	h^{-1}(1-C)\cdot
	\left(1-h^{-1}(1-C)\right)}$ is the Bhattacharyya parameter of a BSC($\BSCprob$) with capacity $C$.
	%
\end{lemma}

\begin{proof}
	We start by defining the log-likelihood ratio:
	\begin{align}
	\Lambda\dfn\ln\left[\frac{\prod_{i=1}^{\rho}P_{Y_i|X=0}(Y_i)}{\prod_{i=1}^{\rho}P_{Y_i|X=1}(Y_i)}\right],
	\end{align}
	and use $\Lambda$ to rewrite the maximum-likelihood decision rule of \eqref{eq:MLdecrule} as: 
	\newtext{
	\begin{align}\label{eq:llrrule}
    \tilde{Y}	=
    \begin{cases}
		 0  &\text{if }   \Lambda>0\\ 
         1  &\text{if }   \Lambda<0\\
         W &\text{if }   \Lambda=0,
	\end{cases}
	\end{align}
	}
	where $W$ is a $\Ber(1/2)$ random variable, drawn independently between uses of $P^{(\rho)}_{\tilde{Y}|\tilde{X}}$.
	Using the sufficient statistic $g(Y_i)=(X\oplus Z_{Ti},T_i)$ from Definition~\ref{def:BMS}, it is easy to show that the log-likelihood function $\Lambda$ can be written as

	The (symmetric) decision error probability can now be upper bounded by
	\begin{align}
	&\delta=\Pr(\tilde{Y}\neq \tilde{X})\\
	&=\Pr(\tilde{Y}\neq \tilde{X}\mid \tilde{X}=0)\\
	&\leq \Pr\left(
	(-1)^{\tilde{X}}\sum_{i=1}^{\rho}(1-2Z_{Ti})\ln\tfrac{1-T_i}{T_i}\leq 0~{\bigg|}~ \tilde{X}=0
	\right)\label{eq:ineq1}
	\\
	&= \Pr\left(\sum_{i=1}^{\rho}(1-2Z_{Ti})\ln\tfrac{1-T_i}{T_i}\leq 0\right).\label{eq:chernofflast}
	\end{align}	
	We note that the inequality in \eqref{eq:ineq1} implies that the event of a \textit{tie} (i.e., $\Lambda=0$) is regarded as an error in probability one, where in fact, due to the random tie breaking, it is an error with probability half.
	We now recall the Chernoff bound for a sum of i.i.d. random variables $A_1,...,A_\rho$:
	\begin{align}
	\Pr\left(\sum_{i=1}^{\rho}A_i\leq {a}\right)\leq e^{s a}\left[\Expt e^{-s A_i}\right]^{\rho}.
	\end{align}
	for any $s>0$. Applying this bound to \eqref{eq:chernofflast} with $A_i=(1-2Z_{Ti})\ln\frac{1-T_i}{T_i}$, $a=0$ and $s=1/2$ yields
	\begin{align}\label{eq:chernoffbatt}
	\delta\leq\batdistanceBMS^\rho
	\end{align} 
	where $\batdistanceBMS$ is defined as the Bhattacharyya parameter of the channel $P_{Y|X}$, which is equal to:	
	\begin{align}
	\batdistanceBMS&= \Expt_{T,Z_T} \left(\left(
	\sqrt{\tfrac{T}{1-T}}\right)^{1-2Z_T}
	\right)\\
	&= \Expt_T \left(\Expt_{Z_T|T}\left(\left(
	\sqrt{\tfrac{T}{1-T}}\right)^{1-2Z_T}~{\bigg|}~ T
	\right)\right)\\
	&= \Expt \left(
	 2\sqrt{T(1-T)}
	\right).\\
	\end{align}
	It was shown by Guill{\'e}n i F{\`a}bregas et. al. \cite{i2012extremes} that among all BMS channels $P_{Y|X}$ with capacity $C$, the Bhattacharyya parameter is maximized by a BSC. Their proof is based on the fact that the function $x\mapsto\sqrt{h^{-1}(x)\cdot\left(1-h^{-1}(x)\right)}$ is concave, and therefore:
	\begin{align}
	&\batdistanceBMS=	\Expt[2\sqrt{T(1-T)}]\\
	&=2\Expt\left[\sqrt{h^{-1}(h(T))\cdot\left(1-h^{-1}(h(T))\right) }\right]\\
	&\leq 2\sqrt{
		 h^{-1}(\Expt\left[h(T)\right])\cdot
		\left(1-h^{-1}(\Expt\left[h(T)\right])\right)}\quad~~
		\label{eq:jensen1}\\
	&=2\sqrt{
		 h^{-1}(1-C)\cdot
		\left(1-h^{-1}(1-C)\right)}\label{eq:delta4}\\
	&=2\sqrt{
		\BSCprob\cdot
		\left(1-\BSCprob\right)}\label{eq:delta5}\\
	&=\batdistance
	\end{align}
	where in \eqref{eq:jensen1} we used Jensen's inequality, in \eqref{eq:delta4} we used the the fact that capacity of a BMS channel is $C=1-\Expt[h(T)]$, and in \eqref{eq:delta5} we used the capacity of the BSC($\BSCprob$) $C=1-h(\BSCprob)$. Combining \eqref{eq:chernoffbatt} and \eqref{eq:delta5} concludes the proof of the lemma.
\end{proof}

We are now ready to prove Theorem~\ref{theorem:interactiveboundBMS}, which is a generalization of Theorem~\ref{theorem:interactivebound} to BMS channels.
\begin{proof}[Proof of Theorem~\ref{theorem:interactiveboundBMS}]
We follow the same lines as the in proof of Lemma~\ref{lemma:gammalemma} and start by converting the BMS($P_{Y|X}$) channel to a BSC($\targetProb'$) with $0<\targetProb'\leq\targetProb$. According to Lemma~\ref{lemma:BMSrepetition} this can be done using 
\begin{align}
\nrepetitions(P_{Y|X}, \delta)\dfn \frac{\log\frac{1}{\targetProb}}
{\log\frac{1}{\batdistance}}+1.
\end{align}
repetitions where $\batdistance=2\sqrt{\BSCprob(1-\BSCprob)}$ is the Bhattacharyya parameter of a BSC($\BSCprob$) with capacity $\Cshannon(\BSCprob)=\Cshannon(P_{Y|X})$ 
. We then apply an interactive coding scheme for the BSC($\targetProb$) with rate $R(\targetProb)$. After normalizing by $\Cinter(P_{Y|X})$ the following bound it obtained:
\begin{align}\label{eq:capratioBMS}
\frac{\Cinter(P_{Y|X})}{\Cshannon(P_{Y|X})}\geq
\frac{R(\delta)}{\nrepetitions(P_{Y|X}, \delta)\Cshannon(P_{Y|X})}. 
\end{align} 
Bounding the denominator of the right hand term in \eqref{eq:capratioBMS}:
\begin{align}
\nrepetitions(P_{Y|X}, \delta)\Cshannon(P_{Y|X})&=
\left(
\frac{\log\frac{1}{\targetProb}}{\log\frac{1}{\batdistance}}+1\right)\Cshannon(P_{Y|X})\\
&=
\left(\frac{\log\frac{1}{\targetProb}}{\log\frac{1}{\batdistance}}+1\right)\Cshannon(\BSCprob),
\end{align}
which is exactly \eqref{eq:lemma4}. The rest of the proof is as in Lemma~\ref{lemma:gammalemma}, and using the same coding scheme to obtain the same numeric value in the lower bound as in Theorem~\ref{theorem:interactivebound}.
\end{proof}

For completeness, we now show that not only $\frac{\Cinter(P_{Y|X})}{\Cshannon(P_{Y|X})}\geq 0.0302$ for any BMS channel, but also the ratio $\frac{\Cinter(P_{Y|X})}{\Cshannon(P_{Y|X})}$ tends to one as the BMS channel becomes cleaner, similarly to the BSC case.
 
\begin{corollary}\label{corollary:ObehaviorForBMS}
For any sequence in $C$ of BMS channels $P_{Y|X}^C$ with $\Cshannon(P_{Y|X}^C)=C$, we have
\begin{align}
\lim_{C\to 1}\frac{\Cinter(P_{Y|X})}{C}=1.
\end{align}
\end{corollary}
\begin{proof}
	We start by proving that without repetitions a BMS($P_{Y|X}$) channel can be reduced to BSC($\BSCprob$) with $\BSCprob\leq\frac{1-\Cshannon(P_{Y|X})}{2}$. As in \cite{hellman1970probability}, the proof is by noting that
	$h(t)\geq 2t$ for any $t\in[0,1/2]$ and therefore:
	\begin{align}
	\Cshannon(P_{Y|X})&=1-\Expt h(T)\\
	&\leq 1-\Expt 2T\\
	&=1-2\BSCprob.
	\end{align}
	The corollary now follows by taking the lower bound for $\RBsc(\BSCprob,\levelblock)$ in Corollary~\ref{corollary:Obehaviour} as a lower bound to $\Cinter(P_{Y|X})$.
\end{proof}

\section{A deterministic coding scheme\label{section:derandomize}}
The coding scheme described throughout this paper uses randomness for two purposes: the randomized polynomial based error detection procedure described in Definition~\ref{def:equalityPrivate}, and the random tie breaking in the repetition decoding described in Lemma~\ref{lemma:repetition} and Lemma~\ref{lemma:BMSrepetition}. In this section we show how the requirements for randomness can be relaxed using a few simple adaptations of the coding scheme. 

\subsection{On the Randomness Requirements of the Error Detection Scheme in Definition~\ref{def:equalityPrivate}\label{subsection:randdetection}}
We start by recalling that the scheme from Definition~\ref{def:equalityPrivate} requires a random generation of a test point $U$ taken from a finite field. We note that original scheme from \cite[p. 30]{communicationComplexityBook} requires only private randomness. Namely, the test point $U$ should be drawn by Alice party and conveyed to Bob. However, so far we assumed that all the test points used by the scheme (denoted by $\mathcal{U}$) are jointly drawn by both parties using a shared random string (i.e., \textit{public randomness}). This choice was made in order to save the communication overhead of conveying the test points from one party to the other, which is prone to reduce the overall rate of the interactive communication scheme.

The first step in modifying the communication scheme to private randomness is showing the number of random test points can be reduced, without affecting the overall rate. We start showing that $|\mathcal{U}|$, the number of random test points required for all the error detections in the interactive coding scheme can be reduced to $o(\protLength)$. This way, if only private randomness is used, $\mathcal{U}$ can be reliably conveyed from one party to the other without affecting the total rate. 
In Subsection~\ref{subsection:extractR} we show how $\mathcal{U}$ can be generated using randomness extracted from the channel, removing the requirement for private randomness.

We start by noting that by construction of error detection scheme, using independently drawn test points for its different actuations, will make their corresponding error mis-detection events statistically independent. It is now in place to discuss the amount of statistical independence required by the coding scheme. In \eqref{eq:PelRecurssionBasic} we assumed that the probability of error mis-detection is independent between layers. That might imply that using $|\mathcal{U}|=\nlevels$ is satisfactory. In fact, if one is concerned only with the average rate of the coding scheme, using only $|\mathcal{U}|=\nlevels$
will lead to the same average rate of Theorem~\ref{theorem:rewindifrateHamming}. 

However, we recall that we defined rate not in the average sense, but rather, we required the reconstruction of the {transcript} with high probability after a predetermined simulation length. To illustrate this delicate difference, consider the example of the one-way 
BEC with feedback. In this example, all the erased bits are retransmitted. So, using the channel $n$ times will result in $n(1-\epsilon)$ bits decoded with zero error, where $\epsilon$ is the erasure probability. This means that the average rate is $1-\epsilon$, which is exactly the Shannon capacity of the BEC($\epsilon$). However, it is interesting to note that since the erasures are drawn i.i.d., for $n\to\infty$ the rate will concentrate around its average and the probability of decoding less than $n(1-\epsilon-\xi)$ bits will vanish in $n$ for any $\xi>0$. This means, that this simple scheme also achieves Shannon's capacity in a stricter deterministic sense - namely, for $n\to\infty$ a number of information bits respective to Shannon's capacity could be reliably transmitted with a vanishing error probability using a fixed number of channel uses. 

For our scheme, the convergence to the average rate is stated in Lemma~\ref{lemma:concnetration}. The concept of the proof appearing in Appendix~\ref{appendix:concentration} is similar to that of the BEC with feedback. We regard the rewind bits as the counterparts of the erasures in the BEC and show that actual number of rewind bits in every layer,  concentrates around its average. A delicate issue in the analysis is the independence of the rewind bits in our scheme. In the first layer, 
the rewind bits are calculated according to Definition~\ref{def:exhammingdetection}. This is a deterministic scheme that is based only on the vectors of channel errors, which are i.i.d between different blocks. Therefore, the rewind bits are indeed i.i.d. For higher layers, the scheme in Definition~\ref{def:equalityPrivate} is used. As explained in the proof of Lemma~\ref{lemma:poly}, the rewind bit is calculating according to 
\begin{align}
\indfunc{\sum_{i=1}^\ell (X_i^A-X_i^B) U^{i-1} (\text{mod } q)\neq 0}.
\end{align}
While it is tempting to assume that $X_i^A-X_i^B$ is exactly the vector of i.i.d channel errors, we note that the ``$-$" operation is done over $\GFq$ and not over $\binarySpace$. This means, that the event of error mis-detection depends not only on the channel error vector, but also on the vectors {related} to the {transcript}: $X_i^A$, $X_i^B$. Since the {transcript} might be dependent between consecutive blocks, the corresponding rewind bits might also be statistically dependent, if the same value of $U$ is used for both blocks. 

One way of breaking this dependence is drawing independent $U$ for every error detection in every layer. As stated before, if common randomness is used, this procedure is feasible, but when using only private randomness it might cause a decrease of the total rate. We recall that in every layer $1<\levelind\leq \nlevels$, there are ${\levelblock^{\nlevels-\levelind}}$ blocks for which error detection is applied using Definition~\ref{def:equalityPrivate}. In our modification of the coding scheme for private randomness we assume that only $\levelblock^{\lceil(\nlevels-\levelind)/2\rceil}$ independent test points are used, such that the test point is changed every 
$\levelblock^{\lfloor(\nlevels-\levelind)/2\rfloor}$ blocks. In Appendix~\ref{appendix:concentration} we prove that this reduced number of independent test points still ensures a slower, yet fast enough, concentration. 

Let us now bound the total number of bits required for the description of $\mathcal{U}$ denoted by $n_U$.  
We recall that the number of bits required for the error detection at layer $1<\levelind\leq\nlevels$ is bounded by $3\levelind\log \levelblock$ by \eqref{eq:errordetectinobits}. 
So, the overall number of bits can be upper bounded by
\begin{align}
n_U&\leq \sum_{\levelind=2}^{\nlevels}3\levelind(\log\levelblock) \levelblock^{\lceil(\nlevels-\levelind)/2\rceil}\\
&\leq
3\levelblock^{\nlevels/2+1}\log\levelblock\sum_{\levelind=2}^{\infty}\levelind \levelblock^{-\levelind/2}\\
&=O(\levelblock^{\nlevels/2}\log\levelblock)\\
&=O(\sqrt{\protLength})
\end{align}
These bits can be conveyed from Alice to Bob before the beginning of the simulation using a block code with some constant positive rate $R_U$ below Shannon's capacity, requiring $\frac{n_U}{R_U}=O(\sqrt{\protLength})$ channel uses. 
However, an error in the decoding of $\mathcal{U}$ might occur, which might cause a failure in the {simulation} of the entire {transcript}. We denote this error event by $\mathcal{E}_3$ and add it to the previously defined error events $\mathcal{E}_1$ and $\mathcal{E}_3$. The probability of $\mathcal{E}_3$ can be upper bounded by an error exponent yielding:
\begin{align}
\Pr(\mathcal{E}_3)\leq e^{-O(\sqrt{\protLength})}
\end{align}
so clearly $\lim_{\protLength\to \infty}\Pr(\mathcal{E}_3)=0$ making this error event negligible.
We should also add $\frac{n_U}{R_U}$ to the total number of channel uses of the scheme in \eqref{eq:errordetoverhead}. But since $\frac{n_U}{R_U}=O(\sqrt{\protLength})$, $\sigmalength$ would change only by $O(\sqrt{\protLength})$, which would not affect the asymptotic value of  the rate from Theorem~\ref{theorem:rewindifrateHamming}.


\subsection{Extracting randomness from the channel \label{subsection:extractR}}
In the previous subsection we showed that the error detection procedure of Definition~\ref{def:equalityPrivate} can be implemented using private randomness requiring $n_U\leq O(\sqrt{\protLength})$ random bits for the entire coding scheme, which were assumed to be drawn by Alice. Our coding scheme can however, be made explicit by extracting the random bits from the channel. 
While a randomness extraction procedure with optimal efficiency was presented by Elias in \cite{elias1972efficient}, we use von-Neumann's suboptimal scheme \cite{vonNeauman1951} due to its simplicity of analysis and the {vanishing} effect of its suboptimality on the total rate.

\begin{lemma}\label{lemma:extractingrandomness}
	The coding scheme can be made explicit by extracting the randomness from the channel with an overhead of 
	\begin{align}
	n_R= O(\sqrt{\protLength})
	\end{align}
	channel uses and an additional error probability
	\begin{align}
	\Pr(\mathcal{E}_4)\leq e^{-O(\sqrt{\protLength})}.
	\end{align}
\end{lemma}
\begin{proof}
	Bob sends Alice $n_R$ zeros and Alice receives a noise vector $Z_1,...,Z_{n_R}$ whose elements are i.i.d Bernoulli($\BSCprob$). 
	Alice then divides the noise elements into pairs. For the pairs $00$ and $11$, Alice does nothing. For the pairs $01$ or $10$ Alice extracts a single random bit valued $0$ or $1$ respectively. Clearly if a bit was extracted, it is $0$ or $1$ with equal probability. 
	We now define $W_i$ as a Bernoulli r.v. that is set to one if a random bit was extracted: 
	\begin{align}
	W_i=\indfunc{Z_{2i-1}Z_{2i}=01 \vee  Z_{2i-1}Z_{2i}=10},
	\end{align}	
	such that $\Pr(W_i=1)=2\BSCprob(1-\BSCprob)$.
	Therefore, the (random) number of extracted bits is
	\begin{align}
	N_R = \sum_{i=1}^{n_R/2} W_i,
	\end{align}
	and the probability of failure in the random bit extraction is
	\begin{align}
	\Pr(\mathcal{E}_4)= \Pr(N_R<n_U).
	\end{align}
	We now set
	\begin{align}
	n_R=\frac{n_U}{\BSCprob(1-\BSCprob)(1-\delta)}=O(\sqrt{\protLength})
	\end{align}
	for some fixed $0<\delta<1$.  
	Using the multiplicative form of Chernoff's bound 
	\begin{align}
	\Pr(\mathcal{E}_4)
	&=\Pr\left(\sum_{i=1}^{n_R/2} W_i<(1-\delta)\Expt\sum_{i=1}^{n_R/2}W_i\right)\\
	&\leq e^{-\frac{\delta^2 n_U}{2(1-\delta)}}\\
	&=e^{-O(\sqrt{\protLength})}
	\end{align}
	
\end{proof}

Using Lemma~\ref{lemma:extractingrandomness}, the explicit scheme that extracts the randomness from the channel has a vanishing error probability with the same rate in as in Theorem~\ref{theorem:rewindifrateHamming}.


\subsection{Treating ties as erasures\label{subsection:ties}\label{subsection:tieerasure}}
We start this discussion by observing a simple example of a  \textit{tie}, which is the erasure event in the BEC. Clearly, if the channel output is an erasure, i.e., $Y=\erasure$, then $\Pr(Y=\erasure\mid X=0)=\Pr(Y=\erasure\mid X=1)$ and a tie occurs. Suppose now, that we would like to adapt the coding scheme of Theorem~\ref{theorem:rewindifrateHamming}, which gives a rate $\RBsc(\delta,\levelblock)$ for a BSC($\delta$), for a BEC($\delta$). Randomly breaking the tie, i.e., uniformly drawing $Y=0$ or $Y=1$ in the case of $Y=\erasure$ will reduce the BEC($\delta$) to a BSC($\delta/2$) and the coding scheme designated for a BSC could be applied. However, we note that the erasure event in the BEC($\delta$) has the same probability of the error event in the BSC($\delta$), which is to be detected in the error detection phase of the rewind-if-error scheme. However, since the erasure is naturally detected by its receiver without requiring an error detection procedure, the rewind-if-error for the BSC could potentially be used, without requiring randomness, and with an improved efficiency.

We can now extend the notion of treating ties as erasures to the general case of a BMS channel. Before we proceed it is instrumental to define  \textit{binary channel with symmetric error and erasure}, BSEC($\delta-\epsilon,\epsilon$), whose transition matrix $P_{Y| X}$ appears in Table~\ref{table:BSEC}.
It is clear from the definition that $\delta\in[0,1/2]$ and $\epsilon\in[0,\delta]$, where $\epsilon=0$ for a BSC($\delta$) and $\epsilon=\delta$ for a BEC($\epsilon$). In addition, it is easy to see that for any 
$\epsilon\in[0,\delta]$, the capacity of the BSEC($\delta-\epsilon,\epsilon$) is
\begin{align}
(1-\epsilon)\left(1-h\left(\frac{1-\delta}{1-\epsilon}\right)\right),
\end{align}
which can be proved by analysis to be strictly larger than $\Cshannon(\delta)$ for every $0<\epsilon\leq \delta$.

\begin{table}\centering
	\begin{tabular}{|c||c|c|c|}\hline
		\backslashbox{$X$}{$Y$}& $0$& $\erasure$ & $1$ \\ \hline\hline
		$0$ & $1-\delta$ & $\epsilon$ & $\delta-\epsilon$ \\ \hline
		$1$ & $\delta-\epsilon$ & $\epsilon$ & $1-\delta$ \\ \hline 
	\end{tabular}\vspace{2mm}
	\caption{The transition matrix $P_{Y| X}$ of a BSEC($\delta-\epsilon,\epsilon$)\label{table:BSEC}}
\end{table}

We now give a non-random version of Definition~\ref{def:rhorepetition} and Lemma~\ref{lemma:BMSrepetition}, in which ties are marked as erasures:
\begin{definition}\label{def:rhorepetitionBEC}[$\rho$-repetition channel with erasures]
	Let $P^{(\rho,\erasure)}_{\tilde{Y}|\tilde{X}}$ be the $\rho$-repetition channel with erasure, corresponding to a BMS($P_{Y|X}$) channel, obtained by transmitting $\rho$ repetitions of the bit $\tilde{X}$ through BMS($P_{Y|X}$) channel and taking 
\newtext{
	\begin{align}
	&\tilde{Y}=\\&
	\begin{cases}
	0 & \text{if} \quad\prod_{i=1}^\rho P_{Y_i|X=0}(Y_i)>\prod_{i=1}^\rho P_{Y_i|X=1}(Y_i)\\ 
	1 & \text{if} \quad\prod_{i=1}^\rho P_{Y_i|X=0}(Y_i)<\prod_{i=1}^\rho P_{Y_i|X=1}(Y_i)\\	
	\erasure & \text{if} \quad\prod_{i=1}^\rho P_{Y_i|X=0}(Y_i)=\prod_{i=1}^\rho P_{Y_i|X=1}(Y_i)\\
	\end{cases}
	\end{align}
}
\end{definition}
\begin{lemma}\label{lemma:BMSrepetitionBEC}
	For any BMS($P_{Y|X}$) channel with Shannon capacity $\Cshannon(P_{Y|X})=C$ the corresponding $\rho$-repetition with erasure channel $P^{(\rho,\erasure)}_{\tilde{Y}|\tilde{X}}$ is a BSEC($\delta-\epsilon,\epsilon$) with 
	$\epsilon\in[0,\delta]$ and 
	$\delta\leq \batdistance^{\rho}$ where $\batdistance$ is as in Lemma~\ref{lemma:BMSrepetition}.
\end{lemma}
\begin{proof}
	The proof follows the same lines as the proof of Lemma~\ref{lemma:BMSrepetition} by making two observations. The first  is by noting that in Definition~\ref{def:BMS} it was implied that an erasure event in a BMS channel corresponds to the statistic $g(Y)=(T,X\oplus Z_T)$ with $T=1/2$ and $X\oplus Z_T$, which is a Bernoulli($1/2$) random bit. In Definition~\ref{def:rhorepetitionBEC}, as well as in the standard BEC definition, such a bit is not produced. However, we note that in the log-likelihood ratio function used for the decision \eqref{eq:llrrule}, the value of the random bit is not used. 
The second observation is by noting that in Lemma~\ref{lemma:BMSrepetition}, ties were pessimistically regarded as errors with probability one, where in fact, the random tie breaking reduces their respective error probability to half. Therefore, marking ties as erasures, the aggregate probability of erasure and error is $\delta$ and the induced channel is a 
BSEC($\delta-\epsilon,\epsilon$) with $\delta$ as in Lemma~\ref{lemma:BMSrepetition} and $\epsilon\in[0,\delta]$.	
\end{proof}

 We are now ready to present the rewind-if-error coding scheme, without tie breaking. We note that ties can appear in two contexts: i) If the original BMS channel had an erasure event (i.e., the probability of $T=1/2$ is strictly positive). ii) If the BMS channel was reduced to BSC using Lemma~\ref{lemma:BMSrepetition} and ties occurred in the decoding. We note that ties cannot occur in the repetition coding used for the transmission of the error detection bits in the BSC scheme, since the number of repetitions is always odd.
 
 For for contexts the rewind-if-error scheme can be modified as follows: when a party receives an erasure, it uses the zero value in order to calculate its next {bit of the transcript}. Then, at the end of the corresponding rewind window, the standard error detection procedure is bypassed and an error is announced. If the erasure was detected by Bob, he simply sets the rewind bit to one and sends it to Alice. If it was detected by Alice, she signals a designated symbol to Bob, indicating the erasure. We note that in the first layer an additional bit was reserved for this purpose. In higher layers, the bound in \eqref{eq:errordetectinobits} ensures that the extra symbol could be signaled without requiring additional bits.

For the sake of completeness, the issue of erasures should also be discussed in the context of randomness extraction in Subsection~\ref{subsection:extractR}. Here, we note that if the channel used for randomness extraction can be reduced to 
a BSEC($\delta-\epsilon,\epsilon$), with $\epsilon<\delta$, Lemma~\ref{lemma:extractingrandomness} could still be used, changing 
$n_R$ only by a constant factor and leaving it in an order of magnitude of $O(\sqrt{\protLength})$. In the extreme case $\epsilon=\delta$ (a \textit{pure} BEC), Lemma~\ref{lemma:extractingrandomness} could not be used. However, in this case all the errors in the scheme in all layers (including the errors of the repetition used for the error detection bits) are marked as erasure. Therefore, 
the random error detection procedure of Lemma~\ref{def:equalityPrivate} need not be used, and random bits need not be extracted from the channel.

\section{Concluding Remarks}
In this paper we revisited the problem of interactive communication over noisy channels originally introduced by Schulman \cite{schulman1992communication}, and studied the problem from an information- and communication-theoretic perspective. We started by defining the interactive channel capacity with respect to a {protocol} and not with respect to a distributed computing problem. As a consequence, our definitions do not use the notion of communication complexity. We then presented a structured and deterministic rewind-if-error coding scheme, and used it to calculate a lower bound for the ratio between the Shannon capacity and the essential interactive capacity of every BMS channel. To the best of our knowledge, this is the first time that a numerical value is {attached} to this ratio. 

We note that the current value of the lower bound can likely be further improved using different coding schemes. A nontrivial upper bound on the ratio between the Shannon capacity and the essential interactive capacity for a fixed channel (i.e., not in the limit of a very clean channel) remains an intriguing open question even in the simplest binary symmetric case. 


\appendices


\section{Proof of Corollary~\ref{corollary:Obehaviour}\label{appendix:Oofh}}
We begin by writing \eqref{eq:hammingrate} as
\begin{align}
\RBsc(\BSCprob,\levelblock)=
\frac{1-A(\BSCprob,\levelblock)}{1+B(\BSCprob,\levelblock)}
\end{align}
where
\begin{align}
&A(\BSCprob,\levelblock)\dfn
\levelblock\BSCprob
+(2+\log\levelblock)\batdistance^{\repetitioncoeffH}\\
&+\frac{\levelblock^2}{\levelblock-1}
\left(
\Pfailure_{1}
+3\batdistance^{\repetitioncoeff+4}\levelblock\log\levelblock
\frac{2-\batdistance^2\levelblock}{(1-\batdistance^2\levelblock)^2}
\right)
\\&
+
3\batdistance^{\repetitioncoeff+4}\levelblock^{2}\log\levelblock
\frac{2-\batdistance^2\levelblock}{(1-\batdistance^2\levelblock)^2}
+\xi
\end{align}
and
\begin{align}
&B(\BSCprob,\levelblock)\dfn
\frac{\repetitioncoeffH(2+\log\levelblock)}{\levelblock}
\\&+3\log\levelblock\left[
\frac{\repetitioncoeff(2\levelblock-1)}{(\levelblock-1)^2}
+\frac{4\levelblock}{(\levelblock-1)^3}
+\frac{4\levelblock-2}{\levelblock(\levelblock-1)^2}
\right]\\&+{o(1)}.
\end{align} 
Using the inequality $1/(1+x)<1-x$ for $x>0$ and the fact that
$A(\BSCprob,\levelblock)\geq 0$, $B(\BSCprob,\levelblock)\geq 0$ gives:
\begin{align}
\RBsc(\BSCprob,\levelblock)&\geq
1-A(\BSCprob,\levelblock)-B(\BSCprob,\levelblock)+A(\BSCprob,\levelblock)B(\BSCprob,\levelblock)\\
&\geq 1-A(\BSCprob,\levelblock)-B(\BSCprob,\levelblock).
\end{align}
We use the definitions $\batdistance=2\sqrt{\BSCprob(1-\BSCprob)}$, $\repetitioncoeff=3$ and $\repetitioncoeffH=5$ and assume from this point on that $\levelblock\to\infty$ and $\BSCprob=o(1/\levelblock)$. Neglecting all high order terms we obtain:
\begin{align}
B(\BSCprob,\levelblock) = O\left(\frac{\log \levelblock}{\levelblock}\right)
\end{align} 
and 
\begin{align}
A(\BSCprob,\levelblock)=\levelblock\BSCprob+O(\levelblock)\Pfailure_{1} +\xi+o(1).
\end{align}
We now recall \eqref{ea:Pe1}
\begin{align}
\Pfailure_{1}&\leq 
\frac{1}{2\levelblock}\left(1+2(\levelblock-1)(1-2\BSCprob)^{\frac{\levelblock}{2}}+(1-2\BSCprob)^\levelblock\right)\\
&-(1-\BSCprob)^\levelblock
+(2+\log\levelblock)\batdistance^{\repetitioncoeffH}\\&=O(\levelblock\BSCprob^2).
\end{align} 
and set 
\begin{align}
\xi=\levelblock^{-2},
\end{align}
which ensures that Lemma~\ref{lemma:concnetration} holds (see \eqref{eq:S1concentration})
obtaining
\begin{align}
&\RBsc(\BSCprob,\levelblock)\geq\\ &1-\left(
\levelblock\BSCprob+O(\levelblock^2\BSCprob^2)+\levelblock^{-2}+o(1)
+O\left(\frac{\log \levelblock}{\levelblock}\right)
\right).
\end{align}
Finally, setting $\BSCprob=\frac{\log\levelblock}{\levelblock^2}$ as in \cite{kol2013interactive} gives
\begin{align}
\RBsc(\BSCprob,\levelblock)&\geq 
 1-O\left(\frac{\log\levelblock}{\levelblock}\right)\\
&=1-O\left(\sqrt{-\BSCprob\log\BSCprob}\right)\\&=1-O\left(\sqrt{h(\BSCprob)}\right). 
\end{align}

\section{Proof of Lemma~\ref{lemma:concnetration}}\label{appendix:concentration}
We would like to prove that
\begin{align}
\lim_{\protLength\to \infty}\Pr(\mathcal{E}_1)=\lim_{\protLength\to \infty}\Pr\left(\cursor^A(\ntrans)< \protLength\right)=0.
\end{align}
We start by recalling \eqref{eq:cursorA}
\begin{align}
\cursor^A(\ntrans)\geq\ntrans\left(1-\sum_{\levelind=1}^{\nlevels}\averewind\right)
\end{align}
The probability of the complementary event is:
\begin{align}\label{eq:appendixAfirstbound}\hspace{-5mm}
\Pr\left(\cursor^A(\ntrans)\geq \protLength\right)\geq 
\Pr\left(1-\sum_{\levelind=1}^{\nlevels}\averewind\geq \frac{\protLength}{\ntrans}\right).
\end{align}
By \eqref{eq:ntransprotlength} we have 
\begin{align}
\frac{\protLength}{\ntrans}=
{1-\sum_{\levelind=1}^\infty \Pblbound-\xi}\leq {1-\sum_{\levelind=1}^\nlevels \Pblbound-\xi},
\end{align}
so we can further bound \eqref{eq:appendixAfirstbound} by
\begin{align}
&\Pr\left(\cursor^A(\ntrans)\geq \protLength\right)\\
&\geq 
\Pr\left(1-\sum_{\levelind=1}^{\nlevels}\averewind\geq 1-\sum_{\levelind=1}^\nlevels \Pbl-\xi\right)\\
&=1-\Pr\left(\sum_{\levelind=1}^\nlevels\averewind>\sum_{\levelind=1}^\nlevels \Pbl +\xi\right).
\label{eq:appendixAsecondbound}
\end{align}
Therefore $\Pr(\mathcal{E}_1)\leq\Pr\left(\sum_{\levelind=1}^\nlevels\averewind>\sum_{\levelind=1}^\nlevels \Pbl +\xi\right)$ and the lemma can be proved by proving
\begin{align}\label{eq:concentrationofbl}
\lim_{\ntrans\to \infty}\Pr\left(\sum_{\levelind=1}^\nlevels\averewind>\sum_{\levelind=1}^\nlevels \Pbl +\xi\right)=0.
\end{align}
We start by observing that
\begin{align}
&\Pr\left(\sum_{\levelind=1}^\nlevels
\averewind
>\sum_{\levelind=1}^\nlevels \Pbl+\xi \right)\\
&\leq
	\Pr\left(\bigcup_{\levelind=1}^\nlevels\left[
	\averewind
	> \Pbl+\frac{\xi}{\nlevels}\right] \right)\\
&\leq
\sum_{\levelind=1}^{\nlevels}\Pr\left(
\averewind
> \Pbl+\frac{\xi}{\nlevels} \right)\\
&=S_1+S_2,
\end{align}
where $S_1\dfn\sum_{\levelind=1}^{\left\lfloor\frac{3}{4}{\nlevels}\right\rfloor}\Pr\left(\averewind> \Pbl+\frac{\xi}{\nlevels} \right)$
and $S_2\dfn\sum_{\levelind={\left\lfloor\frac{3}{4}{\nlevels}\right\rfloor}+1}^{{\nlevels}}\Pr\left(\averewind> \Pbl+\frac{\xi}{\nlevels} \right)$.

Starting with $S_1$, by the definition $\averewind$ in \eqref{eq:blbardef}, the $\levelind$'th summand of $S_1$ is:
\begin{align}\label{eq:S1summand}
&\Pr\left(\averewind> \Pbl+\frac{\xi}{\nlevels} \right)\\&=\Pr\left({\sum_{m=1}^{\levelblock^{\nlevels-\levelind}}\rewindbit^A_{\levelind}(m)}> 
{\levelblock^{\nlevels-\levelind}}(\Pbl+\tfrac{\xi}{\nlevels}) \right).
\end{align}
We recall that $\rewindbit^A_{\levelind}(m)$ are Bernoulli($\Pbl$) r.v.'s with limited independence. The following straightforward generalization of the Chernoff-–Hoeffding Theorem is now useful:
\begin{lemma}\label{lemma:hoefdingextended}
	Let $X_1,...,X_n$ be a series of Bernoulli($p$) r.v.'s, divided into groups of $\ell$ elements. We assume that all distinct groups statistically independent but the r.v.'s within every group might be statistically dependent. Namely, let $i,\tilde{i}\in \{1,...,n/\ell\}$ and $j,\tilde{j}\in\{1,...,\ell\}$. It is given that  $X_{(i-1)\ell+j}$ and $X_{(\tilde{i}-1)+\tilde{j}\ell}$ are statistically independent for every $i\neq\tilde{i}$ and every $j,\tilde{j}$
	but might be statistically dependent for $i=\tilde{i}$ and some $j\neq\tilde{j}$. Then, for every $0<\epsilon<1-p$:
	\begin{align}\label{eq:concentrationLemma}
	\Pr\left(\sum_{i=1}^nX_i\geq n(p+\epsilon)\right)\leq e^{-\frac{n}{\ell} 2\epsilon^2}.
	\end{align}
\end{lemma}

\begin{proof}
	We begin with the standard derivation of the Chernoff bound for $\sum_{i=1}^nX_i$:
	\begin{align}
	&\Pr\left(\sum_{i=1}^nX_i\geq n(p+\epsilon)\right)\\
	&\leq\min_{t>0} e^{-tn(p+\epsilon)}\Expt\left( e^{t\sum_{i=1}^{n} X_{i}}\right)\\	
	&\leq\min_{t>0} e^{-tn(p+\epsilon)}\Expt\left( e^{t\sum_{i=1}^{n/\ell}\sum_{j=1}^{\ell} X_{(i-1)\ell+j}}\right)\\
	&=\min_{t>0} e^{-tn(p+\epsilon)}
	\Expt\left(\prod_{i=1}^{n/\ell}\prod_{j=1}^{\ell}e^{tX_{(i-1)\ell+j}}\right)\\
	&=\min_{t>0} e^{-tn(p+\epsilon)}
	\prod_{i=1}^{n/\ell}
	\Expt\left(\prod_{j=1}^{\ell}e^{tX_{(i-1)\ell+j}}\right)
	\label{eq:trans1}
	\end{align}	
	where in \eqref{eq:trans1} we used the independence assumptions of groups of length $\ell$. We now prove the following  bound for the first group, $i=1$
	\begin{align}\label{eq:HolderIneq}
	\Expt\left(\prod_{j=1}^{\ell}e^{tX_{j}}\right)
	\leq \Expt\left(e^{\ell t X_1}\right).
	\end{align}
	The proof is based on using H\"{o}lder's inequality iteratively. We start by recalling H\"{o}lder's inequality for the expectation of real valued non-negative random variables, $W, V\in\mathbb{R}$, $W, V\geq0$ and $p>1$:
	 \begin{align}\label{eq:holder}
		\Expt (W\cdot V)\leq \left(\Expt\left(W^{\frac{p}{p-1}}\right)\right)^{\frac{p-1}{p}}
		\left(\Expt\left(V^{p}\right)\right)^{\frac{1}{p}}.
	 \end{align}
	Using \eqref{eq:holder} for $\Expt\left(\prod_{j=1}^{\ell}e^{tX_{j}}\right)$ with $W=\prod_{j=1}^{\ell-1}e^{tX_{j}}$, $V=e^{tX_{\ell}}$ and $p=\ell$ gives 	
	\begin{align}
	&\Expt\left(\prod_{j=1}^{\ell}e^{tX_{j}}\right)\\
	&\leq \left(\Expt\prod_{j=1}^{\ell-1}e^{\frac{\ell}{\ell-1}tX_{j}}\right)^{\frac{\ell-1}{\ell}}
	\left(\Expt\left(e^{\ell tX_{\ell}}\right)\right)^{\frac{1}{\ell}}.\label{eq:holderiteration1}
	\end{align}
	Using \eqref{eq:holder} for $\Expt\left(\prod_{j=1}^{\ell-1}e^{\frac{\ell}{\ell-1}tX_{j}}\right)$ with $W=\prod_{j=1}^{\ell-2}e^{\frac{\ell}{\ell-1}tX_{j}}$, $V=e^{\frac{\ell}{\ell-1}tX_{\ell-1}}$ and $p=\ell-1$ gives 	
	\begin{align}
	&\Expt\left(\prod_{j=1}^{\ell-1}e^{\frac{\ell}{\ell-1}tX_{j}}\right)\\
	&\leq 
	\left(\Expt\prod_{j=1}^{\ell-2}e^{\frac{\ell}{\ell-2}tX_{j}}\right)^{\frac{\ell-2}{\ell-1}}
	\left(\Expt\left(e^{\ell tX_{\ell-1}}\right)\right)^{\frac{1}{\ell-1}}.\label{eq:holderiteration2}
	\end{align}
	Plugging \eqref{eq:holderiteration2} into \eqref{eq:holderiteration1} and taking into account that $X_\ell$ and $X_{\ell-1}$ have the same marginal distribution as $X_1$ gives:
	\begin{align}\label{eq:holderlast}
	&\Expt\left(\prod_{j=1}^{\ell}e^{tX_{j}}\right)\\ &\leq 
	\left(\Expt\prod_{j=1}^{\ell-2}e^{\frac{\ell}{\ell-2}tX_{j}}\right)^{\frac{\ell-2}{\ell}}
	\left(\Expt\left(e^{\ell tX_{1}}\right)\right)^{\frac{2}{\ell}}.
	\end{align}
	We now implement this process iteratively on the left hand term the upper bound in \eqref{eq:holderlast} for $p=\ell-2$ to $p=2$ finally giving \eqref{eq:HolderIneq}. 
	
	We now notice that \eqref{eq:HolderIneq} depends only on the marginal distribution of a single sample, which is assumed to be Bernoulli($p$), so it should hold for all groups $i\in \{1,...,n/\ell\}$. Therefore we can use \eqref{eq:HolderIneq} for all the elements in the outer product in \eqref{eq:trans1} giving:
	\begin{align}
	&\Pr\left(\sum_{i=1}^nX_i\geq n(p+\epsilon)\right)\\
	&\leq\min_{t>0} e^{-tn(p+\epsilon)}\left(\Expt\left( e^{t\ell X_1}\right)\right)^{n/\ell}\label{eq:trans2}\\
	&\leq\left(\min_{t\ell>0} e^{-t\ell (p+\epsilon)}\Expt\left( e^{t\ell X_1}\right)\right)^{n/\ell}\\
	&=e^{-\frac{n}{\ell}d_n\left((p+\epsilon)||p\right)}.\label{eq:chermin}
	\end{align} 
	where \eqref{eq:chermin} is by the standard minimization of the Chernoff bound and 
	$d_n(p||q)\dfn p\ln\frac{p}{q}+(1-p)\ln\frac{1-p}{1-q}$
	is Kullback-Leibler Divergence between two Bernoulli random variable with probabilities $p$ and $q$, which is now calculated with respect to the natural logarithm basis.
	Finally, by Pinsker's inequality we bound the divergence by $d_n\left(p+\epsilon||p\right)\geq 2\epsilon^2$ and obtain \eqref{eq:concentrationLemma}.
\end{proof}
We can now use Lemma~\ref{lemma:hoefdingextended} to bound \eqref{eq:S1summand}. Recalling the discussion from Subsection~\ref{subsection:randdetection}, at every layer $1<\levelind\leq \nlevels$, there are ${\levelblock^{\nlevels-\levelind}}$ blocks for which error detection is applied using Definition~\ref{def:equalityPrivate}. We assume that only $\levelblock^{\lceil(\nlevels-\levelind)/2\rceil}$ independent test points are used, which are changed every 
$\levelblock^{\lfloor(\nlevels-\levelind)/2\rfloor}$ blocks. So, we can use 
Lemma~\ref{lemma:hoefdingextended} on \eqref{eq:S1summand} where the number of independent groups is
$\frac{n}{\ell}=\levelblock^{\lceil(\nlevels-\levelind)/2\rceil}$ yielding:
\begin{align}
&\Pr\left({\sum_{m=1}^{\levelblock^{\nlevels-\levelind}}\rewindbit^A_{\levelind}(m)}> 
{\levelblock^{\nlevels-\levelind}}(\Pbl+\tfrac{\xi}{\nlevels}) \right)\\&\leq
e^{-\levelblock^{\lceil(\nlevels-\levelind)/2\rceil}\tfrac{2\xi^2}{\nlevels^2}}\\
&\leq
e^{-\levelblock^{(\nlevels-\levelind)/2}\tfrac{2\xi^2}{\nlevels^2}}\label{eq:S1summandbound}
\end{align}
Summing all the element is of $S_1$ yields:
\begin{align}\label{eq:S1concentration}\hspace{-5mm}
S_1\leq 
\sum_{\levelind=1}^{\left\lfloor\frac{3}{4}{\nlevels}\right\rfloor}
e^{-\levelblock^{(\nlevels-\levelind)/2}\tfrac{2\xi^2}{\nlevels^2}}
\leq \frac{3}{4}{\nlevels}\cdot
e^{-\levelblock^{\nlevels/8}\tfrac{2\xi^2}{\nlevels^2}}.
\end{align}
The second transition is by using the maximal summand obtained at $\levelind=\left\lfloor\frac{3}{4}{\nlevels}\right\rfloor$.
Recalling that $\nlevels=\log_\levelblock\ntrans$, it is clear that 
 $\lim_{\ntrans\to\infty}S_1=\lim_{\nlevels\to\infty}S_1=0$.

 Proceeding with $S_2$:
\begin{align}\label{eq:S2bound}
S_2&=\sum_{\levelind=\left\lfloor\frac{3}{4}{\nlevels}\right\rfloor+1}^{{\nlevels}}\Pr\left(\averewind> \Pbl+\frac{\xi}{\nlevels} \right)\\
&\leq \sum_{\levelind=\left\lfloor \frac{3}{4}{\nlevels}\right\rfloor+1}^{{\nlevels}}\Pr\left(\averewind> 0 \right).
\end{align}
Observe that if $\averewind> 0$ then at least one rewind bit at level $\levelind$ is set to one. So, we can use the union bound and obtain 
\begin{align}\label{eq:pblconcentration}
\Pr\left(\averewind> 0 \right)\leq \levelblock^{\nlevels-\levelind}\Pbl.
\end{align}
Recalling \eqref{eq:PblboundH} 

\begin{align}
\Pbl&\leq 
\levelblock^{2-\levelind}
\left(
\levelblock\Pfailure_{1}
+3\batdistance^{\repetitioncoeff+4}\levelblock^{2}\log\levelblock
\frac{2-\batdistance^2\levelblock}{(1-\batdistance^2\levelblock)^2}
\right)\\
&+3\batdistance^{\repetitioncoeff}(\log\levelblock)\levelind\batdistance^{2\levelind}
\end{align}
we can further bound \eqref{eq:pblconcentration} by
\begin{align}\label{eq:Pblbarbound}
&\Pr\left(\averewind> 0 \right)\leq \\
&\levelblock^{\nlevels}
\left(
\levelblock^{-2\levelind}
\left(
\levelblock^3\Pfailure_{1}
+3\batdistance^{\repetitioncoeff+4}\levelblock^{4}\log\levelblock
\frac{2-\batdistance^2\levelblock}{(1-\batdistance^2\levelblock)^2}
\right)\right.\\
&\left.+3\batdistance^{\repetitioncoeff}(\log\levelblock)\levelind
\left(\tfrac{\batdistance^2}{\levelblock}\right)^{\levelind}
\right).
\end{align}
Observing that the bound in \eqref{eq:Pblbarbound} is monotonically decreasing in $\levelind$ for a sufficiently large $\levelind$ we can bound the summands of $S_2$ by the term obtained at $\levelind=3/4\nlevels$, yielding:
\begin{align}
S_2&\leq \frac{\nlevels}{4}
\levelblock^{-\nlevels/2}
\left(
\levelblock^3\Pfailure_{1}
+3\batdistance^{\repetitioncoeff+4}\levelblock^{4}\log\levelblock
\frac{2-\batdistance^2\levelblock}{(1-\batdistance^2\levelblock)^2}
\right)\\&+
\frac{9}{16}
\batdistance^{\repetitioncoeff}(\log\levelblock)
\nlevels^2
\left(\left(\tfrac{\batdistance^2}{\levelblock}\right)^{3/4}\levelblock\right)^{\nlevels}.
\end{align}
It is clear that the left hand term is monotonically decreasing in $\nlevels$. Analyzing the right hand term, we use the definition of $\batdistance$ and we observe that 
\begin{align}
\left(\tfrac{\batdistance^2}{\levelblock}\right)^{3/4}\levelblock
=&\left(\batdistance^6\levelblock\right)^{1/4}
<\left(
2^6\BSCprob^3\levelblock
\right)^{1/4}\\
&<(2^6/(8\levelblock)^3\levelblock)^{1/4}
<(2^3\levelblock^2)^{-1/4}<1
\end{align}
where the third transition is due to the assumption that
$\BSCprob<1/(8\levelblock)$ in Theorem~\ref{theorem:rewindifrateHamming}. All in all, setting $\nlevels=\log_\levelblock\ntrans$ guarantees that $\lim_{\ntrans\to\infty}S_2=0$, which concludes the proof of Lemma~\ref{lemma:concnetration}.

\bibliographystyle{IEEEtran}

\end{document}